\documentclass[a4paper,12pt,eqno]{article}
\usepackage{latexsym,amssymb,amsfonts,amsmath,amsthm}
\usepackage{graphicx}
\def\qed{\hfill $\Box$}

\newtheorem{thm}{Theorem}[section]
\newtheorem{lem}[thm]{Lemma}
\newtheorem{cor}[thm]{Corollary}
\newtheorem{pro}[thm]{Proposition}
\newtheorem{rem}[thm]{Remark}

\newtheorem{definition}{Definition}

\newcommand{\bs}[1]{\boldsymbol{#1}}

\title{Quantum walks driven by quantum coins with two multiple eigenvalues}
\author{
{\small Norio Konno}\\
{\scriptsize Department of Applied Mathematics, Faculty of Engineering, Yokohama National University}\\
{\scriptsize Hodogaya, Yokohama 240-8501, Japan}\\
{\small Iwao Sato}\\
{\scriptsize Oyama National College of Technology}\\
{\scriptsize Oyama, Tochigi 323-0806, Japan}\\
{\small Etsuo Segawa\footnote{corresponding author: e-mail: segawa-etsuo-tb@ynu.ac.jp} }\\
{\scriptsize Graduate School of Environment and  Information Sciences}\\
{\scriptsize Yokohama National University}\\
{\scriptsize Hodogaya, Yokohama 240-8501, Japan}\\
{\small Yutaka Shikano } \\
{\scriptsize Graduate School of Science and Technology, Gunma University} \\ 
{\scriptsize Maebashi, Gunma 371-8510, Japan} \\
{\scriptsize Institute for Quantum Studies, Chapman University} \\
{\scriptsize Orange, CA 92866, USA} \\
{\scriptsize JST PRESTO} \\
{\scriptsize Kawaguchi, Saitama 332-0012, Japan}
}
\date{}
\begin{document}
\maketitle

\par\noindent
\begin{small}
\par\noindent
{\bf Abstract}. We consider a spectral analysis on the quantum walks on graph $G=(V,E)$ with the local coin operators $\{C_u\}_{u\in V}$ and the flip flop shift. 
The quantum coin operators have commonly two distinct eigenvalues $\kappa,\kappa'$ and  $p=\dim(\ker(\kappa-C_u))$ for any $u\in V$ with $1\leq p\leq \delta(G)$, where $\delta(G)$ is the minimum degrees of $G$.  
We show that this quantum walk can be decomposed into a cellular automaton on $\ell^2(V;\mathbb{C}^p)$ whose time evolution is described by a self adjoint operator $T$ and its remainder. We obtain how the eigenvalues and its eigenspace of $T$ are lifted up to as those of the original quantum walk. As an application, we express the eigenpolynomial of the Grover walk on $\mathbb{Z}^d$ with the moving shift in the Fourier space. 
  
\footnote[0]{
{\it Keywords: } 
Quantum walk, Spectral mapping theorem  
}
\end{small}

\section{Introduction}
Spectral information of the Grover walk, Szegedy walk~\cite{Sze}, and staggered walk on connected graphs can be reduced to that of the underlying random walks and cyclic information of the graphs (see \cite{MOS, HPSS} and its reference therein). These types of quantum walks can be expressed by the graph zeta of the Hashimoto type and can be deformed to that of the Ihara type~\cite{EmmsETAL2006,KS}. Some graph structures are extracted considering the support of the time evolution operator from the zeta expression~\cite{HKSS,KSS,RenETAL}.

We refer to the family of quantum walks whose spectral information can be reduced to that of an underlying cellular automaton, whose transition weights are given by complex numbers on a vertex set as Ihara's class. For underlying cellular automatons that are random walks, some interesting behaviours of such quantum walks in Ihara's class have been mathematically revealed. For example, the  efficiency of the quantum walk in quantum search algorithms (see \cite{Por} and its references) and periodicity~\cite{HKSS2, Yo}, a type of perfect state transfer~\cite{AdaH} of the walk and a characterisation of the edge state using recurrence properties of a random walk~\cite{IKS} have been investigated. However, Ihara's class is only one among a wide variety of classes of quantum walks. Indeed, the time evolution operator of a quantum walk on a graph is determined by the choice of the local quantum coin assigned at each vertex $u$ which is a unitary matrix on the $d(u)$-dimensional space and the shift operator which is a permutation on the arcs satisfying $t(a)=o(a)$. Here, $d(u)$ is the degree of vertex $u$ and,  $t(a)$ and $o(a)$ are the terminal and origin vertices of the arc $a$, respectively. The shift operator of the transposition on the arcs is called a flip-flop shift. Particularly, if the shift operator is a flip-flop, and we choose quantum coins such that  
\[ \mathrm{Spec}(C_u)\subset\{\kappa,\kappa'\},\;\dim\ker(\kappa-C_u)=1, \text{ for any $u\in V$},\]
where $\kappa,\kappa'\in \mathbb{C}$ with $|\kappa|=|\kappa'|=1$ are independent of $u$, then the quantum walk is included in Ihara's class. 
Here, $\mathrm{Spec}(X)$ is the set of eigenvalues of a square matrix $X$. 
In contrast, the Grover walk with the {\it moving shift} on $\mathbb{Z}^d$ which seems to be a natural setting, is not included in the Ihara class for $d\geq 3$ because, as described in greater detail herein, the dimensionalities of $\ker(\kappa -C_u)$ and $\ker(\kappa'-C_u)$ are greater than $2$ when $d\geq 3$ after rewriting the time evolution to describe the shift operator as a flip-flop shift operator. Recently, interesting aspects of quantum walks not included in Ihara's class have been reported \cite{KT,MNJ}.  In this paper, we construct a class such that not only the Grover walk and Szegedy walk but also quantum walks on $\mathbb{Z}^d$ with both the flip-flop shift and moving shift types are included. We show how the properties obtained by previous studies~\cite{MOS,HPSS} hold and how they are deformed in this class. This small extension provides a new motivation for investigating a Hermitian matrix on $\ell^2(V;\mathbb{C}^p)$, where $p$ is a parameter of the extension model. Particularly, if $p=1$, then Ihara's class is reproduced.    
As an application, we obtain the eigenpolynomial of the Grover walk on $\mathbb{Z}^d$ with the moving shift type. 

This paper is organized as follows. 
In Section~2, we explain that all the coined quantum walk can be regarded as a quantum walk with a flip-flop shift by considering the local permutation operator to the local coin operators. 
In Section~3, we propose the quantum walk model considered. 
We provide some properties of the boundary operators as well as some examples of the discriminant operator $T$ whose eigenvalues are raised to the unit circle in the complex plane with a one-to-two mapping rule. The discriminant operator is symmetric and represents a walk with a matrix-valued weight. 
In Section~4, we present an example of the Grover walk on $\mathbb{Z}^d$ with a moving shift. We demonstrate that the eigenpolynomial of this quantum walk can be obtained exactly using our main theorem. 
In Section~5, we present the spectral information comprising our main results. First, we show an outline of how the eigenvalues are raised to the unit circle as the eigenvalues of the time evolution operator by using the zeta function method. Second, we refine this mapping theorem and show how the eigenspace of $T$ is raised as the eigenspace of $U$. Finally, we provide a summary and discussion in the final Section.

\section{Definition of several quantum walks on a graph} 

\subsection{Notation of graphs} 

Graphs treated here are finite. 
Let $G=(V(G),E(G))$ be a connected graph 
with the set $V(G)$ of vertices and the set $E(G)$ of unoriented edges $uv$ 
joining two vertices $u$ and $v$. 
Two vertices $u$ and $v$ of $G$ are {\em adjacent} if there exits an edge $e$ joining $u$ and $v$ in $G$. 
Furthermore, two vertices $u$ and $v$ of $G$ are {\em incident} to $e$.  
The {\em degree} $\deg v = \deg {}_G \  v$ of a vertex $v$ of $G$ is the number of edges incident to $v$. 
For a natural number $k$, a graph $G$ is called {\em $k$-regular } if $\deg {}_G \  v=k$ for each vertex $v$ of $G$. 

For $uv \in E(G)$, an arc $(u,v)$ is the oriented edge from $u$ to $v$. 
Set $A(G)= \{ (u,v),(v,u) | uv \in E(G) \} $. 
For $a=(u,v) \in A(G)$, set $u=o(a)$ and $v=t(a)$. 
Furthermore, let $\bar{a}=a^{-1}=(v,u)$ be the {\em inverse} of $a=(u,v)$. 
A {\em path} $P=(v_1, v_2 , \ldots , v_{n+1} )$ {\em of length $n$} in $G$ is a sequence of $(n+1)$ vertices such that 
$v_i v_{i+1} \in E(G)$ for $i=1, \ldots ,n$. 
Then $P$ is called a {\em $(v_1 , v_{n+1} )$-path}. 
If $e_i =v_i v_{i+1} (1 \leq i \leq n)$, then we write $P=( e_1 , \ldots e_n )$.

\subsection{Quantum walk on a graph} 
For a discrete set $\Omega$, the vector space whose standard basis is labelled by each element of $\Omega$ is denoted by $\mathbb{C}^{\Omega}$. 
The standard basis of $\mathbb{C}^{\Omega}$ are described by 
    \[\delta_\omega^{(\Omega)}(\omega')=\begin{cases} 1 & \text{: $\omega=\omega'$}\\ 0 & \text{: otherwise,} \end{cases}\]
for any $\omega\in \Omega$.
Let us set a permutation $\pi:A(G)\to A(G)$ on the symmetric arc set $A:=A(G)$, satisfying $o(\pi(a))=t(a)$ for any $a\in A(G)$. If the permutation $\pi$ fulfils $\pi(a)=a^{-1}$, for any $a\in A$, then we call $\pi_0:=\pi$ a flip-flop permutation. 
We set $\mathcal{A}:=\mathbb{C}^A$. 
Let $S_\pi: \mathcal{A}\to \mathcal{A}$ represent the permutation $\pi$ by 
 \[ S_\pi\delta_a^{(A)}=\delta_{\pi(a)} \]
for any $a\in A$. 
We call $S_\pi$ a shift operator of $\pi$. 
Particularly, if $\pi$ is the flip flop permutation, we call $S_0:=S_{\pi_0}$ a flip-flop shift operator. 

For a vertex $u\in V$, the subset $A_u\subset A(G)$, in which all terminal vertices are commonly $u\in V$, is denoted by 
    \[ A_u=\{ a\in A \;|\; t(a)=u\}. \]
We set a unitary operator on $\mathbb{C}^{A_u}$ by $C_u$, which is represented by a $d(u)\times d(u)$ unitary matrix. To extend the domain of $C_u$ to the entire space $\mathcal{A}$, let us introduce $\chi_u:\mathbb{C}^{A}\to \mathbb{C}^{A_u}$ as 
\[ (\chi_u\psi)(a) = \psi(a) \]
for any $a\in A_u$. Note that the adjoint of $\chi_u$ is expressed as  
    \[ (\chi_u^*\psi)(a) = \begin{cases}\psi(a) & \text{: $a\in A_u$}\\ 0 & \text{: otherwise.}\end{cases} \]
Because $A(G)$ can be described by the disjoint union $\sqcup_{u\in V}A_u$, the entire space $\mathbb{C}^{A}$ can be decomposed into  $\mathcal{A}= \oplus_{u\in V} (\chi_u^*\mathbb{C}^{A_u}\chi_u)$.
Under this decomposition, the coin operator $C$ on $A(G)$ is defined by 
    \[ C=\bigoplus_{u\in V} (\chi_u^*C_u\chi_u).  \]
\begin{definition}
The time evolution operator of the quantum walk on $G=(V,E)$ with the permutation $\pi$ on $A(G)$ and the sequence of local coin operators $(C_u)_{u\in V}$ is defined by $U=U(\pi,(C_u)_{u\in V})=S_\pi C$.
\end{definition}
The time iteration of the quantum walk on $G$ is described by $\psi_{n+1}=U\psi_n$ with some initial state $\psi_0\in \mathcal{A}$. 
Note that by the definition of $U$, if $t(b)\neq o(a)$, then 
    \[ \langle \delta^{(A)}_a,U\delta^{(A)}_b\rangle =0 \]
holds, and $\chi_u\chi_u^*=I_{A_u}$. Note that there are many possibilities for the choice of such a $\pi$; indeed, we have $\prod_{u\in V}d(u)!$ choices. 
However, any time evolution operator $U_\pi=S_\pi C$ can be rewritten by a time evolution operator with a flip-flop shift as follows.
\begin{pro}
For any permutation $\pi$ on $A=A(G)$ satisfying $t(a)=o(\pi(a))$ and the coin operator $C=\oplus_{u\in V}C_u$, we have 
\[U(\pi,(C_u)_{u\in V})=U(\pi_0,(C_u')_{u\in V}).\]  
Here $C'_u=Q_u(\pi)C_u$ with $Q_u(\pi)\delta_a^{(A_u)}=\delta_{(\pi(a))^{-1}}^{(A_u)}$. 
\end{pro}
\begin{proof}
Because $S_0^2=I$, we have 
    \begin{align*}
        U_\pi &= S_\pi C =S_0 (S_0 S_\pi C). 
    \end{align*}
The composition of the permutations $\pi_0\circ \pi$ is 
\[ a \stackrel{\pi}{\mapsto} \pi(a) \stackrel{\pi_0}{\mapsto} (\pi(a))^{-1}.  \]
Because the permutation $\pi$ satisfies with $t(a)=o(\pi(a))$, we have $t(a)=t(\pi(a))^{-1}$. 
This means that the composition $\pi_0\circ \pi_m$ is a local permutation on the same terminal vertex. 
This implies that $S_0S_\pi$ may be decomposed into $\oplus_{u\in V}(\chi_u^*Q_u(\pi)\chi_u)$. 
Here $Q_u(\pi)\delta_{a}^{(A_u)}=\delta_{(\pi(a))^{-1}}^{(A_u)}$. Therefore we can regard $S_0S_\pi C=\oplus_{u\in V} \chi_u^*(Q_u(\pi)C_u)\chi_u$ as the coin operator.  
\end{proof}

The Grover walk is a special case of a quantum walk performed by choosing the flip-flop shift operator and the Grover matrix $C_u=(2/d(u) )J_u-I_u$ as the coin operators. Here, $J_u$ and $I_u$ are the all-ones matrix and the identity matrix on $\mathbb{C}^{A_u}$.  More precisely, ${\bf U} ={\bf U} (G)=( U_{ef} )_{e,f \in D(G)} $, 
of $G$ is defined by 
\[
U_{ef} =\left\{
\begin{array}{ll}
2/d_{t(f)} (=2/d_{o(e)} ) & \mbox{if $t(f)=o(e)$ and $f \neq e^{-1} $, } \\
2/d_{t(f)} -1 & \mbox{if $f= e^{-1} $, } \\
0 & \mbox{otherwise.}
\end{array}
\right. 
\]

\section{The quantum walk treated in this paper } 

\subsection{Motivation}\label{sect:motivation1}
In the above section, we see that for any time evolution operator of quantum walks,  applying appropriate permutation to the row vectors of each local coin operator, we can reproduce the original time evolution by the flip flop shift time evolution.  
The spectrum mapping theorem of quantum walk is quite useful to see the spectral information of a special class of quantum walks. 
The spectrum of such a quantum walk is 
generated by the fundamental cycles of the graph, and  
inherited by the spectrum of a self adjoint operator $T$ on $\mathbb{C}^V$, which is,  under some condition, isomorphic to  transition matrix of a reversible random walk on the graph.   
The quantum walks, which can be applied the traditional spectrum mapping theorem, have commonly the following property:
\begin{enumerate}
    \item $\mathrm{Spec}(C_u)\subseteq \{1,-1\}$ ;
    \item $\dim\ker(1-C_u)= 1$.
\end{enumerate}

In this paper, we extend this condition by
\begin{enumerate}
    \item $\mathrm{Spec}(C_u)\subseteq \{\kappa,\kappa'\}$ ;
    \item $\dim\ker(\kappa-C_u)= p$.
\end{enumerate}
Here $\kappa\neq \kappa'$ to avoid a trivial walk. 

To explain the motivation of our extended setting, 
let us consider the Grover walk with the moving shift on $\mathbb{Z}^d$. The time evolution operator is $U_m=S_mC$, where $C=\oplus_{u\in V}\mathrm{Gr}(2d)$, where $\mathrm{Gr}(2d)$ is the $2d$-dimensional Grover matrix. To explain the shift operator $S_m$, let us prepare the notation. 
Let $e_j\in \mathbb{Z}^d$ be the standard basis of $\mathbb{R}^d$ ($j=1,\dots,d$). The arc whose terminal vertex is $x\in \mathbb{Z}^d$ and the origin is $x\mp e_j$ is denoted by $(x;\pm j)$. The permutation of the moving shift $S_m$ is defined by 
\[ \pi_m(x;\pm j)=(x\pm e_j;\pm j). \]
We set the moving shift operator by $S_m:=S_{\pi_m}$.
On the other hand, the permutation of the flip flop shift $S_0$ is defined by 
\[ \pi_0(x;\pm j)=(x\mp e_j;\mp j). \]
We set the flip flop shift operator by $S_0:=S_{\pi_0}$. 
The composition $\pi_0\circ\pi_m$ is 
\[ (x;\pm j) \stackrel{\pi_m}{\mapsto} (x\pm e_j;\pm j) \stackrel{\pi_0}{\mapsto} (x;\mp j). \]
Let the permutation matrix $\sigma:=P_x$ to the coin operator be the transposition $(x;j)\leftrightarrow (x;-j)$ $j=1,\dots,d$.  
Thus putting $C'=\oplus_{u\in V} (\sigma \cdot \mathrm{Gr}(2d))$, we have 
\[ U_m= S_0C'. \]
Let us compute the spectrum of $C'_d:=\sigma\cdot \mathrm{Gr}(2d)$. 
It is easy to see that $C'_d \bs{u}_d=\bs{u}_d$, where $\bs{u}_d$ is the uniform vector in $\mathbb{C}^{2d}$. 
Let the computational basis of $C_x=C_d'$ are labeled by $(x;1),(x;-1),\dots,(x;d),(x;-d)$ by this order. 
If $\bs{v}=\oplus_{j=1}^d\bs{v}_j\in \mathbb{C}^{2d}$ with some $2$-dimensional vectors $\bs{v}_j$ ($j=1,\dots,d$) is orthogonal to $\bs{u}_d$, then we have $C'_d\bs{v}=-\sigma \bs{v}$. 
Since $\sigma$ is the transposition, this permutation matrix can be decomposed into  $\sigma=\oplus_{j=1}^d\sigma_1$, where $\sigma_1$ is the Pauli matrix. If $\sigma_1\bs{v}_j=\bs{v}_j$ for any $j=1,\dots,d$, then we have $C'_d\bs{v}=-\bs{v}$.
We can construct such eigenbasis so that they are orthogonal to $\bs{u}_d$ by 
\[ \left\{ [1,\omega,\dots,\omega^{d-1}]^\top \otimes 
\begin{bmatrix}1\\ 1 \end{bmatrix} \;:\; \omega^d=1,\;\omega\neq 1  \right\}. \]
The dimension is $d-1$. 
On the other hand, if $\bs{v}_j=-\sigma\bs{v}_j$ for any $j=1,\dots,d$, we have $C'_d\bs{v}=\bs{v}$. 
We can construct such eigenbasis by 
\[ \left\{ \bs{e}_j\otimes \begin{bmatrix}1\\ -1 \end{bmatrix} \;:\; j=1,\dots,d  \right\}, \]
where $\bs{e}_j\in \mathbb{C}^d$ ($j=1,\dots,d$) are the standard basis of $\mathbb{C}^d$.  
The dimension is $d$.
We summarize the above statement about the spectrum of $C'_d$ in the following. 
\begin{lem}\label{lem:Grover}
Let $C'_d=\sigma\cdot \mathrm{Gr}(2d)$. Then we have 
\begin{enumerate}
    \item $\mathrm{Spec}(C'_d)=\{ 1,-1 \}$,
    \item $\dim(\ker(1-C'_d))=d+1$, $\dim(\ker(1+C'_d))=d-1$.
\end{enumerate}
\end{lem}
When $d=1$, the walk becomes trivially zigzag walking, because $C_d'$ is the identity matrix. 
When $d=2$, we see the condition  $\dim(\kappa-C_u)=1$ is satisfied by putting $\kappa=-1$ and $\kappa'=+1$, the traditional spectral mapping theorem can be applied just multiplying $(-1)$ to the entire time evolution operator $U$. 
On the other hand, when $d\geq 3$, the condition for the traditional spectral mapping is broken. 
Konno and Takahashi~\cite{KT} show that the eigenspace which causes the localization of the Grover walk with the moving shift on $\mathbb{Z}^d$ is generated by the unit of each $d$-dimensional hypercube, while the one with the flip flop shift is generated by every quadrangle cycle. Then the number of arcs which constructs the unit of  generator of the birth eigenvector is greater than the one with the flip flop shift. 
Indeed, the number of arcs for the moving shift is $d\times 2^d$ while the one for the flip flop shift is $4d$. 
In such a natural setting of the moving shift, 
obtaining the spectral information is still open from the view point of the spectral mapping theorem, then in this paper, we extended condition imposing to the coin operator and consider the eigenpolynomial of the time evolution operator.  

\subsection{Our setting}
Let $\delta(G)$  be the minimum degree of $G$. 
In this paper, fixing a natural number $1\leq  p \leq  \delta(G)$ and distinct unit complex numbers $\kappa$ and $\kappa'$, we relax the condition imposing to the coin operator by 
\begin{enumerate}
    \item $\mathrm{Spec}(C_u)\subseteq \{\kappa, \kappa'\}$ for any $u\in V$;
    \item $\dim\ker(\kappa-C_u)= p$,
\end{enumerate}
and consider the following flip flop shift operator; that is, $(S\psi)(a)=\psi(\bar{a})$. 
The time evolution operator is $U=SC$. 

We set $\{\alpha_u^{(j)}\}_{j=1}^{p}$ as a completely orthogonal normalized system (CONS) of $\ker(\kappa-C_u)$. Let $w(\cdot): A\to \mathbb{C}^p$ be defined by
    \[ w(a)= 
    \begin{bmatrix}  {\alpha_{t(a)}^{(1)}}(a)^* \\ \vdots \\ {\alpha_{t(a)}^{(p)}}(a)^*\end{bmatrix} . \]
Using this, we set a matrix valued weight associated with the motion of a walker moving along arcs by $W: A\to M_p(\mathbb{C})$ so that $W(a):=w(a)w(\bar{a})^*$, that is, 
    \[ (W(a))_{i,j} = \overline{\alpha_{t(a)}^{(i)}(a)}\;\;\alpha_{o(a)}^{(j)}(\bar{a}),\;\;(i,j\in\{1,\dots,p\}). \]
Let $\mathcal{V}_p:=\mathbb{C}^V\otimes \mathbb{C}^p$. 
For any $\vec{f},\vec{g}\in \mathcal{V}_p$, the inner product is $\langle \vec{g},\vec{f} \rangle_{\mathcal{V}_p}=\sum_{u\in V}\langle \vec{g}(u),\vec{f(u)} \rangle_{\mathbb{C}^p}$. 
\begin{definition}
The operator $T$ on $\mathcal{V}_p$ is defined by
\[ (T)_{u,v}:=(\delta_u^{(V)} \otimes I_p)^*\; T\; (\delta_v^{(V)}\otimes I_p )=\sum_{o(a)=u,\;t(a)=v} W(a)  \]
for any $u,v\in V$. 
We call $T$ the discriminant operator. 
\end{definition}
Note that the summation shows an existence of the  multi-edges between $u$ and $v$.\footnote{This will be useful to consider the quotient graph on the crystal lattice in the Fourier space, just changing the shift operator $S^{(\theta)}\delta_a=e^{i\theta(a)}\delta_{\bar{a}}$ by a one-form function $\theta(\bar{a})=-\theta(a)$. }
The $u,v$ element of $T$ describes the $p$-dimensional matrix valued weight associated with the motion of a walker from $v$ to $u$. 
Therefore 
\[ (T^n)_{u,v}=\sum_{\begin{matrix}(a_1,\dots,a_n)\in A^{n} \text{ such that}\\ 
v=o(a_1),\; t(a_1)=o(a_2),\dots,t(a_{n-1})=o(a_n),\; t(a_n)=u
\end{matrix}}W(a_n)\dots W(a_1) \]
for any $n\geq 1$. Note that the $(u,v)$ element of $T$ is a $p$-dimensional matrix. 
Especially, when $p=1$, the discriminant operator is reduced to  $W(a)=\overline{\alpha_{t(a)}(a)}\alpha_{o(a)}(\bar{a})\in \mathbb{C}$. Let us put $\eta\in \mathbb{C}^A$ by $\eta(a):=\alpha_{t(a)}(a)$ for $p=1$ case. If there exists $\pi\in \mathbb{C}^V$ such that  $\eta(a)\pi(t(a))=\eta(\bar{a})\pi(o(a))$, then $T=DPD^{-1}$ with $(Df)(u)=\pi(u)f(u)$ and $P=[P_{u,v}]_{u,v\in V}= \sum_{t(a)=u,o(a)=v}|\eta(a)|^2 $ which is a  reversible probability transition matrix. 

Let $\tilde{\alpha}_u^{(j)}$ be the extension of  $\alpha_{u}^{(j)}\in \mathbb{C}^{A_u}$ to $\mathbb{C}^A$ such that 
    \[ \tilde{\alpha}_{u}^{(j)}(a) = \begin{cases}
    \alpha_u^{(j)}(a) & \text{: $t(a)=u$} \\
    0 & \text{: otherwise.}
    \end{cases} \] 
Let $\tilde{w}(a)$ be the extension of $w(a)\in \mathbb{C}^p$ to  $\mathcal{V}_p$ such that 
    \[ (\tilde{w}(a))(u)=\begin{cases}
    w(a) & \text{: $t(a)=u$,} \\
    0 & \text{: otherwise.}
    \end{cases} \]
When the sets of vertices and arcs $V=\{u_1,\dots,u_n\}$ and $A=\{a_1,\dots,a_m\}$, the boundary operator $K:  \mathcal{V}_p \to \mathcal{A}$ is represented by the following matrix form:
    \[ K:=[\tilde{\alpha}^{(1)}_{u_1}\cdots \tilde{\alpha}^{(p)}_{u_1}\;\cdots\;\tilde{\alpha}^{(1)}_{u_n}\cdots \tilde{\alpha}^{(p)}_{u_n}]= \begin{bmatrix}\tilde{w}(a_1)^* \\ \vdots \\ \tilde{w}(a_m)^* \end{bmatrix}
    \]
which is a $|A|\times p|V|$ matrix. 
An equivalent expression of $K$ is 
    \[ (Kf)(a) = \langle \; w(a), f(t(a))\; \rangle_{\mathbb{C}^p},  \]
for any $f\in \mathcal{V}_p$ and $a\in A$. 
The adjoint operator of $K$ is expressed by 
    \[ (K^*\psi)(u) = \sum_{t(a)=u} \psi(a)w(a) \]
for any $\psi\in \mathcal{A}$ and $u\in V$.
Then we have the important properties of $K$. 
\begin{lem}\label{lem:123}
Let $K$, $S$, $T$, $C$ be defined as the above. We have
\begin{align}
I_{\mathcal{V}_p} &= K^* K  \\
T &= K^* S K \\
C &= (\kappa- \kappa') KK^*+\kappa' I_p
\end{align}
\end{lem}
\begin{proof}
For the part (1),  we have 
\begin{align*}
    (K^* K)_{u,v} 
    &= \begin{bmatrix} {\tilde{\alpha}_u}^{(1)\;*}  \\ \vdots \\ {\tilde{\alpha}_u}^{(p)\;*}\end{bmatrix}
    \begin{bmatrix}{\tilde{\alpha}_v}^{(1)} & \cdots & {\tilde{\alpha}_v}^{(p)}\end{bmatrix}
    = \delta_{u,v}I_p.
\end{align*}
For the part (2), 
note that 
\[K(\delta_u^{(V)}\otimes I_p)=[(\tilde{w}(a_1))(u),\dots,(\tilde{w}(a_m))(u) ]^\top = \sum_{t(b)=u}\delta_b^{(A)}\otimes w(b)^*, \]
while
\[ (\delta_v^{(V)}\otimes I_p)^* K^*=\sum_{t(a)=v}{\delta_a^{(A)}}^* \otimes w(a), \]
and $(S)_{a,b}=\langle \delta_a^{(A)},S\delta_b^{(A)} \rangle=\delta_{b,\bar{a}}$. Then we have 
\begin{align*}
    (K^* SK)_{u,v} &= \sum_{t(a)=u,\;t(b)=v}  ({\delta_a^{(A)}}^* \otimes w(a)) (S\otimes 1) (\delta_b^{(A)}\otimes w(b)^*) \\
    &= \sum_{t(a)=u,\;t(b)=v}(S)_{a,b}\; w(a)w(b)^*  \\
    &= \sum_{t(a)=u,\;o(a)=v} w(a) w(\bar{a})^* \\
    &= \sum_{t(a)=u,\;o(a)=v}W(a).  
\end{align*}
For the proof of (3), 
since $\mathrm{Spec} (C_u ) =\{ \kappa,\kappa' \} $, it can be written by 
\[
C_u =(\kappa-\kappa') \sum^{p}_{j=1}  \alpha^{(j)}_u {\alpha^{(j)}_u}^*  +\kappa' I.
\]
Because the coin operator $C$ is 
\[
C= \oplus {}_{u \in V} C_u ,
\]
we only need to show that $KK^*$ coincides with the projection operator 
\[ \sum_{u\in V}\sum^{p}_{j=1}  \tilde{\alpha}^{(j)}_u {\tilde{\alpha}^{(j)}_u\;*}. \]
But this is immediately obtained by the definition of $K$.
\end{proof}
\section{Examples: expression of $T$ for the Grover walk with moving shift type on $\mathbb{Z}^d$ ($d\geq 3$)}\label{sect:example}
In this section, we give a matrix expression of the discriminant operator $T$. As we will see later in the next Section~\ref{sect:main}, the eigenvalues of $T$ gives the main part of the those of $U$ by the mapping given by Corollary~\ref{cor:smt}, see also Fig.~\ref{fig:1}. Then first we demonstrate an application of our main theorem Theorem~\ref{thm:main2} to the Grover walk with moving shift type on $\mathbb{Z}^d$. In particular, we show the eigenpolynomial of the time evolution operator in the Fourier space.

Let the set of arcs of $\mathbb{Z}^d$ be \[A=\{(\bs{x},\epsilon j) \;|\; \bs{x}\in\mathbb{Z}^d,\epsilon\in\{\pm1 \},\;j=1,\dots,d\}.\]
The arc $(\bs{x},\epsilon j)$ represents the arc whose terminal vertex is $\bs{x}$ and the origin vertex is $\bs{x}-\epsilon\bs{e}_j$. Here $\bs{e}_j$ is the standard basis of $\mathbb{R}^d$. The quantum coin assigned at each vertex is an operator on $\mathbb{C}^{2d}$ whose computational basis are labeled by $\{1,\dots,2d\}$. So we need to determine a labeling $\zeta_u:\{a\in A \;|\; t(a)=\bs{x}\}\to \{1,\dots,2d\}$ at each vertex which are bijection map. Such a labeling way can be considered innumerably, but in this paper we fix the map by 
    \[ \zeta(\bs{x},\epsilon j):=\zeta_u(\bs{x},\epsilon j) =2j-\frac{1}{2}(1+\epsilon) \]
for any $u\in \mathbb{Z}^d$ which seems to be ``natural". 
It is possible to see both (i) $(\kappa,\kappa')=(-1,1)$ and (ii) $(\kappa,\kappa')=(1,-1)$.  
As we have discussed in Section~\ref{sect:motivation1}, we start the consideration on the moving shift type Grover walk from $U=SC'$, where $S$ is the flip flop shift and $C'=\oplus_{x\in \mathbb{Z}^d}\sigma \mathrm{Gr}(2d)$. 
\subsection{Fourier transform}
The arc set of $\mathbb{Z}^d$ is isomorphic to $\mathbb{Z}^d \times \{\epsilon j \;|\; \epsilon\in \{\pm 1\},\;j=1,\dots,d\}$. Putting $A_o:=\{\epsilon j \;|\; \epsilon\in \{\pm 1\},\;j=1,\dots,d\}$ and $\mathbb{T}=[0,2\pi)$, we define the Fourier transform $\mathcal{F}:\ell^2(\mathbb{Z}^d\times A_o)\to L^2(\mathbb{T}^d\times A_o)$ by 
\[ (\mathcal{F}\psi)(\bs{k},\epsilon j)=\sum_{\bs{x}\in \mathbb{Z}^d} \psi(\bs{x},\epsilon j)e^{i\langle \bs{x},\bs{k} \rangle} \]
for any $\psi\in \ell^2(\mathbb{Z}^d \times A_o)$ and $(\bs{k},\epsilon j)\in \mathbb{T}^d\times A_o$.
Note that $\ell^2(\mathbb{Z}^d \times A_o)$ is the Hilbert space of the Grover walk on $\mathbb{Z}^d$ and the space $L^2(\mathbb{T}^d\times A_o)$ is its Fourier space. 
The inverse Fourier transform is described by 
\[ (\mathcal{F}^{-1}\hat{f})(\bs{x},\epsilon j)=\int_{\bs{k}\in \mathbb{T}^d} \hat{f}(\bs{k},a)e^{-i\langle \bs{k},\bs{x}\rangle} \frac{d\bs{k}}{(2\pi)^d} \]
for any $\hat{f}\in L^2(\mathbb{T}^d\times A_o)$ and $(\bs{k},\epsilon j)\in \mathbb{T}^d\times A_o$. 
Then we obtain that  \[(\mathcal{F}S\mathcal{F}^{-1}\hat{f})(\bs{k},\epsilon j)=e^{i\epsilon k_j}\hat{f}(\bs{k},-\epsilon j) \text{ and } (\mathcal{F}C\mathcal{F}^{-1}\hat{f})(\bs{k},\epsilon j)=\sum_{\ell=1}^{2d}(C_d')_{2j-(1+\epsilon)/2,\;\ell}\;\hat{f}(\bs{k},\ell).\] 
The time evolution operator in the Fourier space is denoted by $\mathcal{F}U\mathcal{F}^{-1}$.  
Putting $\hat{S}(\bs{k})$ as a unitary operator on $\mathbb{C}^{A_o}$ by $(\hat{S}(\bs{k})g)(\epsilon j)=e^{i\epsilon k_j}g(-\epsilon j)$ and  \[\hat{f}(\bs{k}):=[\hat{f}(\bs{k},+1),\hat{f}(\bs{k},-1),\dots,\hat{f}(\bs{k},+d),\hat{f}(\bs{k},-d)]^\top\in \mathbb{C}^{A_o},\] 
we obtain
\[ [\mathcal{F}U\mathcal{F}^{-1}\hat{f}](\bs{k},\epsilon j)=[ \hat{S}(\bs{k})C_d'\hat{f}(\bs{k}) ](2j-(1+\epsilon)/2). \]
We put $\hat{U}(\bs{k}):=\hat{S}(\bs{k})C_d'$. This is the time evolution operator of the twisted quantum walk on the $d$-bouquet graph with the one-form $\theta(\epsilon j)=e^{i\epsilon k_j}$, ($\epsilon j\in A_o$). Since there is only one vertex $o$ in the $d$-bouquet graph, we have $\mathcal{V}_p=\mathbb{C}^{\{o\}}\times \mathbb{C}^p\cong \mathbb{C}^p$ for this twisted quantum walk. See Subsection~\ref{subsect:oneform}.  Let CONS of $\ker(-1-C_d')$ be $\{\alpha_j\}_{j=1}^{d-1}$ while CONS of $\ker(1-C_d')$ be $\{\bs{\beta_j}\}_{j=1}^{d+1}$. 
If $(\kappa,\kappa')=(-1,1)$, then $p=d-1$ and 
the boundary operator for the twisted quantum walk $K_o:\mathbb{C}^{d-1} \to \mathbb{C}^{A_o}$ is reduced to 
\[ K_o=[\alpha_1,\dots,\alpha_{d-1}], \]
while $(\kappa,\kappa')=(1,-1)$, then $p=d+1$ and 
$K_o:\mathbb{C}^{d+1} \to \mathbb{C}^{A_o}$ is reduced to 
\[ K_o=[\beta_1,\dots,\beta_{d+1}]. \]
The discriminant operator of the twisted quantum walk $\hat{T}(\bs{k})$ is expressed by $\hat{T}(\bs{k})=K_o^*\hat{S}(\bs{k})K_o$.    
Let $ \mathcal{U}: \ell^2(\mathbb{Z}^d)\to L^2([0,2\pi))$ be the Fourier transform such that 
\[ (\mathcal{U}f)(\bs{k})=\sum_{\bs{x}\in \mathbb{Z}^d} f(\bs{x})e^{i\langle \bs{x},\bs{k} \rangle}. \]
The inverse is expressed by 
\[ (\mathcal{U}^{-1}\hat{f})(\bs{x})=\int_{\bs{k}\in \mathbb{T}^d} \hat{f}(\bs{k})\frac{d\bs{k}}{(2\pi)^d}.  \]
Note that $K_o=\mathcal{F}K\mathcal{U}^{-1}$ since $\sum_{x\in \mathbb{Z}^d}e^{i\langle \bs{x},  \bs{k}\rangle}/(2\pi)^d=\delta(\bs{k})$.
Then this twisted discriminant operator can be expressed by $\hat{T}(\bs{k})\hat{f}(\bs{k})=(\mathcal{U}T\mathcal{U}^{-1}\hat{f})(\bs{k})$ because
\begin{align*}
    (\mathcal{U}T\mathcal{U}^{-1}\hat{f})(\bs{k})
    &=(\mathcal{U}K^*\mathcal{F}^{-1}(\mathcal{F}S\mathcal{F}^{-1}) \mathcal{F}K\mathcal{U}^{-1}\hat{f})(\bs{k}) \\
    &= K_o^* \hat{S}(\bs{k}) K_o \hat{f}(\bs{k}) \\
    &= \hat{T}(\bs{k})\hat{f}(\bs{k}). 
\end{align*}
By the argument of Subsection~\ref{subsect:oneform}, the spectral information for the twisted quantum walk can be reproduced by that of the spectral information of $\bs{k}=\bs{0}$ in Theorem~\ref{thm:main2} by changing $S=\hat{S}(\bs{0})$ to $\hat{S}(\bs{k})$ and $T=\hat{T}(\bs{0})$ to $\hat{T}(\bs{k})$. 
\subsection{(i) $(\kappa,\kappa')=(-1,1)$ case: }
By Lemma~\ref{lem:Grover}, CONS of $\ker(-1-C'_{d})$, whose dimension is $d-1$, can be described by
\begin{multline*} \frac{1}{\sqrt{2d}}\;[1,\omega,\dots,\omega^{d-1}]^\top\otimes[1\;1]^\top,\;\frac{1}{\sqrt{2d}}\;[1,\omega^2,\dots,\omega^{2(d-1)}]^\top\otimes[1\;1]^\top,\dots\\
\dots,\frac{1}{\sqrt{2d}}\;[1,\omega^{d-1},\dots,\omega^{(d-1)(d-1)}]^\top\otimes[1\;1]^\top, 
\end{multline*}
where $\omega=e^{2\pi i/d}$. 
Therefore the vector $w(\bs{x};\epsilon j)\in \mathbb{C}^{d-1}$ is obtained as follows: 
\begin{equation}\label{eq:omomi}
    w(\bs{x};\epsilon j) = \frac{1}{\sqrt{2d}}\; [\omega^{-(j-1)},\omega^{-2(j-1)},\dots,\omega^{-(d-1)(j-1)} ]^\top
\end{equation}
for $\epsilon\in\{\pm 1\}$ and $j\in\{1,\dots,d\}$. Due to such a labeling, we have $w(\bar{a})=w(a)$. 
Then the matrix weight associated with the moving along the arc $a=(\bs{x},\epsilon j)$ is expressed by 
    \begin{align*} 
    W(\bs{x};\epsilon j) &=  w(a)w(\bar{a})^* \\
    &= \frac{1}{2d}
    \begin{bmatrix}
    1 & \omega^{j-1} & \omega^{2(j-1)} & \dots & \omega^{(d-1)(j-1)} \\
    \omega^{-(j-1)} & 1 & \omega^{j-1} & \dots & \omega^{(d-2)(j-1)} \\
    \omega^{-2(j-1)} & \omega^{-(j-1)} & 1 & \dots & \omega^{(d-3)(j-1)} \\
    \vdots           &\vdots & \vdots & \vdots & \vdots \\  
    \omega^{-(d-1)(j-1)} & \omega^{-(d-2)(j-1)} & \omega^{-(d-3)(j-1)} & \dots & 1
    \end{bmatrix}.
    \end{align*}
For example, for $d=3$, letting $W_{\epsilon x}:=W(\bs{x},\epsilon \bs{e}_x)$, $W_{\epsilon y}:=W(\bs{x},\epsilon \bs{e}_y)$ and $W_{\epsilon z}=W(\bs{x},\epsilon \bs{e}_z)$ $(\epsilon\in\{\pm 1\},\;\bs{x}\in\mathbb{Z}^3)$ be the matrix weights associated with the moving $x$, $y$, $z$ directions, respectively, then we have \[ W_x=W_{-x}= \frac{1}{6}\begin{bmatrix}1 & 1 \\ 1 & 1 \end{bmatrix},\;W_y=W_{-y}=\frac{1}{6}\begin{bmatrix}1 & \omega \\ \omega^{-1} & 1\end{bmatrix},\; W_z=W_{-z}=\frac{1}{6}\begin{bmatrix} 1 & \omega^{-1} \\ \omega & 1 \end{bmatrix}. \]
In the Fourier space, $T$ is deformed by 
\begin{align*}
     \hat{T}(k_x,k_y,k_z)
     &=e^{ik_x}W_x+e^{-ik_x}W_x+e^{ik_y}W_y+e^{-iky}W_y+e^{ikz}W_z+e^{-ikz}W_z\\
     &= 2(\cos k_x W_x+\cos k_y W_y +\cos k_z W_z) 
\end{align*}
for any $k_x,k_y,k_z\in[0,2\pi)$.
Then the solutions of $\det(\mu-\hat{T}) =0$   
are equivalent to the ones of the following quadratic equation: 
\begin{multline}\label{eq:EPT}
    \mu^2-\frac{2}{3}(\cos k_x+\cos k_y +\cos k_z)\mu+\frac{1}{3}(\cos k_x\cos k_y+\cos k_y\cos k_z+\cos k_z\cos k_x)=0. 
\end{multline}
For example, if we consider the $d$-dimensional torus with the size $N$, taking $k_j=2\pi\ell_j/N$ ($j=1,2,3$), we have the eigenvalues of $T$. 
According to our result in Theorem~\ref{thm:main2}, the eigenvalues of the time evolution operator $\hat{U}(k_x,k_y,k_z)$ in the Fourier space restricted to  $\mathcal{L}$ are lifted up to the unit circle as follows: 
\[ \{-e^{\pm i\arccos (\mu)} \;|\; \mu\in \mathrm{Spec}(\hat{T}(k_x,k_y,k_z))\}. \]
Indeed, the eigeneploynomial of $\hat{U}(k_x,k_y,k_z)$ is computed in \cite{KKS} by 
\[ (1-\lambda^2)\left(\lambda^4+\frac{4}{3}\gamma_1\lambda^3+(2+\frac{4}{3}\gamma_2)\lambda^2+\frac{4}{3}\gamma_1\lambda+1\right), \]
where $\gamma_1=\cos k_x+\cos k_y+\cos k_z$ and $\gamma_2=\cos k_x\cos k_y+\cos k_y\cos k_z+\cos k_z\cos k_x$. 
Let us see the second term in the above polynomial.
Taking the product of $\lambda^{-2}$ and switching the signature of $\lambda$; that is, $\lambda\leftrightarrow -\lambda$, the eigenpolynomial in (\ref{eq:EPT}) is reproduced when we put $\mu=(\lambda+\lambda^{-1})/2$. 
The switching the signature of $\lambda$ derives from the spectral map (\ref{eq:mapping}). The reason for the switching the signature derives from $\kappa=-1$.  

Finally, as a by product of considering this example, we obtain the following property of $T$ under a special condition, which seems to correspond to so called  double stochastic matrix or reversible matrix although we need more discussions. 
\begin{pro}
Assume $G$ be a connected $2d$-regular graph. 
Let the quantum walk be the Grover walk with the $p$-multiplicity of eigenvalue $-1$ and $w\in  \mathbb{C}^p$ be the vector from the CONS of the eigenspace of $-1$ described by the RHS of (\ref{eq:omomi}). 
Let the labeling of arcs $\zeta$ at each vertex satisfy $w(a)=w(\bar{a})$ for any $a\in A$. Then we have 
\begin{equation}
    \sum_{u\in V}(T)_{u,v}=\sum_{v\in V}(T)_{u,v}=I_p.
\end{equation}
\end{pro}
\begin{proof}
It is enough to show that 
    \begin{equation}
        \sum_{o(a)=u}W(a) = \sum_{t(a)=u}W(a)= I_p.
    \end{equation}
First, let us see $\sum_{o(a)=u}W(a)=\sum_{t(a)=u}W(a)$. 
    \begin{align*}
        \sum_{o(a)=u} W(a) 
    &= \sum_{o(a)=u} w(a)w(\bar{a})^* \\
    &= \sum_{o(a)=u} w(\bar{a})w(a)^* \\
    &= \sum_{t(a)=u} w(a)w(\bar{a})^* = \sum_{t(a)=u}W(a).
    \end{align*}
Here the second equality derives from the assumption $w(\bar{a})=w(a)$. 
By a direct computation, we have  
\begin{align*}
    \sum_{t(a)=u} W(a) 
    &= \sum_{t(a)=u} w(a)w(\bar{a})^* \\
    &= \sum_{t(a)=u}w(a)w(a)^* =\sum_{t(a)=u} \left[\; \overline{\alpha_{u}^{(i)}(a)}\; \alpha_{u}^{(j)}(a)\;\right]_{i,j=1,\dots,p}  \\
    &= \left[\; \langle \alpha_u^{(j)},\alpha_u^{(i)} \rangle\;\right]_{i,j=1,\dots,p} = \left[\; \delta_{ij}\;\right]_{i,j=1,\dots,p} \\
    &= I_p
\end{align*}
Here we used the assumption of $w(a)$ in the second equality and the fifth equality derives from the orthogonormality of $\alpha_u^{i}$. 
\end{proof}
This property conserves the following quantity :  
    \begin{align*} 
    \sum_{u\in V}(Tf)(u) &= \sum_{u\in V}\sum_{t(a)=u}W(a)f(o(a)) 
    = \sum_{u\in V}\sum_{o(a)=u}W(a)f(o(a)) \\
    &= \sum_{u\in V}f(u).
    \end{align*}
This means $\sum_{u\in V}f_j(u)=\sum_{u\in V}(Tf)_j(u)$ for any $j\in 1,\dots,p$, where $g(u)=[g_1(u),\dots,g_p(u)]^\top$ for $g\in \{f , (Tf)(u)\}$.  
\subsubsection{(ii) $(\kappa,\kappa')=(1,-1)$ case}
By Lemma~\ref{lem:Grover}, CONS of $\ker(1-C'_d)$, whose dimension is $d-1$, can be described by
\begin{equation*} 
\frac{1}{\sqrt{2d}}[1,\dots,1]^\top,
\frac{1}{\sqrt{2}}\;|1\rangle\otimes [1,\;-1]^\top,\;\frac{1}{\sqrt{2}}\;|2\rangle\otimes [1,\;-1]^\top,\dots,\frac{1}{\sqrt{2}}\;|d\rangle \otimes [1,\;1]^\top,
\end{equation*}
where $|j\rangle=\delta_j\in\mathbb{C}^d$. 
Then 
\begin{align*}
    w(\bs{x},1) &= \left[\frac{1}{\sqrt{2d}},\frac{1}{\sqrt{2}},0\dots,0\right]^\top,\;
    w(\bs{x},-1) = \left[\frac{1}{\sqrt{2d}},\frac{-1}{\sqrt{2}},0\dots,0\right]^\top, \\
    w(\bs{x},2) &= \left[\frac{1}{\sqrt{2d}},0,\frac{1}{\sqrt{2}},\dots,0\right]^\top,\;
    w(\bs{x},-2) = \left[\frac{1}{\sqrt{2d}},0,\frac{-1}{\sqrt{2}},\dots,0\right]^\top, \\
    \vdots &  \\
    w(\bs{x},d) &= \left[\frac{1}{\sqrt{2d}},0,0,\dots,\frac{1}{\sqrt{2}}\right]^\top,\;
    w(\bs{x},-d) = \left[\frac{1}{\sqrt{2d}},0,0,\dots,\frac{-1}{\sqrt{2}}\right]^\top.
\end{align*}
Since the inverse arc of $(\bs{x},\epsilon j)$ is $(\bs{x}-\epsilon j,-\epsilon j)$ and the quantum coin is uniformly assigned, 
the matrix weight associated with the motion  along the arc $a=(\bs{x},\epsilon j)$ is 
    \begin{align}
        (W(\bs{x},\epsilon j))_{\ell,m}&= w(\bs{x},\epsilon j)w(\bs{x}-\epsilon\bs{e}_{j},-\epsilon j)^*
        =w(\bs{x},\epsilon j)w(\bs{x},-\epsilon j)^* \notag \\
        &=
        \begin{cases}
        1/(2d) & \text{: $\ell=m=1$,}\\
        -\epsilon/\sqrt{4d} & \text{: $\ell=1$, $m=j+1$,}\\
        \epsilon/\sqrt{4d} & \text{: $\ell=j+1$, $m=1$,} \\
        -1/2 & \text{: $\ell=m=j+1$,} \\
        0 & \text{: otherwise.}
        \end{cases}
    \end{align}
For example, in $d=3$ case, letting $W_{\pm x}:=W(\bs{x},\pm 1)$, $W_{\pm y}:=W(\bs{x},\pm 2)$, $W_{\pm z}:=W(\bs{x},\pm 3)$, we have 
    \begin{align*}
        W_x & = W_{-x}^*=
        \begin{bmatrix}
        1/6 & -1/\sqrt{12} & 0 & 0 \\
        1/\sqrt{12} & -1/2 & 0 & 0 \\
        0 & 0 & 0 & 0 \\
        0 & 0 & 0 & 0 
        \end{bmatrix}, \\
        W_y & = W_{-y}^*=
        \begin{bmatrix}
        1/6 & 0 & -1/\sqrt{12} & 0 \\
        0 & 0 & 0 & 0 \\
        1/\sqrt{12} & 0 & -1/2 & 0 \\
        0 & 0 & 0 & 0 
        \end{bmatrix}, \\
        W_z & = W_{-z}^*=
        \begin{bmatrix}
        1/6 & 0 & 0 & -1/\sqrt{12} \\
        0 & 0 & 0 & 0 \\
        0 & 0 & 0 & 0 \\
        1/\sqrt{12} & 0 & 0 & -1/2 
        \end{bmatrix}.
    \end{align*}
Then, in the Fourier space, $T$ is deformed by 
    \[ \hat{T}(k_x,k_y,k_z)=
    \begin{bmatrix}
    \frac{1}{3}(\cos k_x+\cos k_y+\cos k_z) & \frac{-i}{\sqrt{3}}\sin k_x
 & \frac{-i}{\sqrt{3}}\sin k_y & \frac{-i}{\sqrt{3}}\sin k_z \\
 \frac{i}{\sqrt{3}}\sin k_x & -\cos k_x & 0 & 0 \\
 \frac{i}{\sqrt{3}}\sin k_y & 0 & -\cos k_y  & 0 \\
 \frac{i}{\sqrt{3}}\sin k_z & 0 & 0 & -\cos k_z
 \end{bmatrix}. \]
The $\mathbb{Z}^3$ lattice is the abelian covering of a $3$-bouquet graph. 
Then the dimension of the Fourier space of the time evolution operator $\hat{U}(k_x,k_y,k_z)$ coincides with the number of arcs of a $3$-bouquet graph; that is, $6$.   
We obtain the $4$ eigenvalues of $\hat{T}(k_x,k_y,k_z)$ by 
\begin{multline*}
\mathrm{Spec}(\hat{T}(k_x,k_y,k_z))=\bigg\{\pm 1,\; \\ \frac{-1}{3}(c [k_x]+c[k_y]+c[k_z])\pm  \frac{\sqrt{2}}{6}\sqrt{3+c[2k_x]+c[ 2k_y]+c[2k_z] -2(c[k_x]c[k_y]+2c[k_y]c[k_z]+c[k_z]c[k_x])} \bigg\}.
\end{multline*}
Here $c[k]=\cos k$.
As we will see, such eigenvalues are lifted up to the unit circle in the complex plain as the eigenvalues of $\hat{U}$ by the one-to-two (\ref{eq:mapping}) except $\pm 1$ while  
the $\pm 1$ eigenvalues are lifted up by a one-to-one map by Theorem~\ref{thm:main2}. 
Then all of the $6$ eigenvalues of the time evolution operator $\hat{U}(k_x,k_y,k_z)$ are directly obtained by  $\mathrm{Spec}(\hat{T}(k_x,k_y,k_z))$. 


From the above observation to the $d=3$ case, we can give more general argument. To this end, let us prepare some notations. 
Let $\Gamma:=\{-\frac{1}{d}\sum_{j=1}^d \cos k_j, \cos k_1,\dots, \cos k_d \}$ and $\Gamma^{(0,s)}:=\Gamma\setminus\{ -\frac{1}{d}\sum_{j=1}^d \cos k_j, \cos k_s \}$ $(1\leq s\leq d)$. 
Moreover for a subset $\Omega\subset \Gamma$, we define 
    \[ \gamma_j(\Omega):=\sum_{\Omega'\subset \Omega,\;|\Omega'|=j}\;\;\prod_{\omega\in \Omega'}\omega \]
for $j=0,1,\dots,|\Omega|$. Here we define  $\gamma_0(\Omega)=0$. 
Then we obtain the following proposition.
\begin{pro}\label{pro:ddim}
The real parts of the eigenvalues of the Grover walk on $\mathbb{Z}^d$ with the moving shift in the Fourier space for fixed $\bs{k}\in \mathbb{T}^d$ are the roots of the following polynomial. 
\begin{equation}\label{eq:epT}
    P(x)=x^{d+1}+\eta_1x^d+\cdots+\eta_d x+\eta_{d+1}, 
\end{equation}
where each coefficient is described by 
\begin{align}
    \eta_1 &= \gamma_1(\Gamma), \\
    \eta_j &= \gamma_j(\Gamma)-\frac{1}{d}\sum_{s=1}^d  \gamma_{j-2}(\Gamma^{(0,s)})\;\sin^2 k_s,\;\;(2\leq j \leq d+1)\;. 
\end{align}
In particular, $P(+ 1)=P(-1)=0$, and its muliplicities are simple for any $\bs{k}=(k_1,\dots,k_d)\in \mathbb{T}^d$ except the wave numbers satisfying   
\[ \#\{1\leq j \leq d\;|\; \cos k_j=\mp 1 \} \geq 2, \] respectively. 
\end{pro}
\begin{proof}
The discriminant operator in the Fourier space $\hat{T}(\bs{k})$ $(\bs{k}=(k_1,\dots,k_d))$ is expressed by 
\begin{equation}\label{eq:dom}
    \hat{T}(\bs{k})=
    \begin{bmatrix}
    \frac{1}{d}\sum_{j=1}^d \cos k_j & \frac{-i}{\sqrt{d}} \sin k_1 & \frac{-i}{\sqrt{d}}\sin k_2 & \cdots & \frac{-i}{\sqrt{d}}\sin k_d \\
    \frac{i}{\sqrt{d}}\sin k_1 & -\cos k_1 & 0 & \cdots & 0 \\
    \frac{i}{\sqrt{d}}\sin k_2 & 0 & -\cos k_2 & \cdots & 0 \\
    \vdots & \vdots & \vdots & \ddots & \vdots \\
    \frac{i}{\sqrt{d}}\sin k_d & 0 & 0 & \dots & -\cos k_d
    \end{bmatrix}.
\end{equation}
This is a selfadjoint operator on the star graph with self loops whose center vertex is labeled by $0$ and the leaves are labeled by $1,\dots,d$. 
Let us put $a_0=(1/d)\sum_{j=1}^d \cos k_j$, $c_j=-\cos k_j$ $(j=1,\dots,d)$ and  $D=\mathrm{diag}(c_1,\dots,c_d)$ which is a diagonal matrix, and $\bs{v}=[(i/\sqrt{d})\sin k_1,\dots,(i/\sqrt{d})\sin k_d]^\top$. 
Then we have
\begin{align}\label{eq:ep}
\det(x-\hat{T}(\bs{k})) &= \det\left( \begin{bmatrix} x-a_0 & -\bs{v}^* \\ -\bs{v} & x I-D \end{bmatrix} \right) \notag\\
&= \det\left(\begin{bmatrix} 1 & -\bs{v}^* \\ 0 & x I-D \end{bmatrix}\begin{bmatrix} (x-a_0)-\bs{v}^* (x I-D)^{-1}\bs{v} & 0 \\ -(x I-D)^{-1}\bs{v} & I \end{bmatrix}\right) \notag\\
&=\det(x I-D)\; \left(x-a_0-\langle \bs{v}, (x I-D)^{-1}\bs{v}\rangle\right) \notag\\
&= (x-c_1)\cdots(x-c_d)\left\{x- a_0-\frac{\sin^2 k_1}{d} \frac{1}{x-c_1}-\cdots-\frac{\sin^2 k_d}{d} \frac{1}{x-c_d} \right\} \notag\\
&= (x-a_0)(x-c_1)\cdots (x-c_d) - \frac{1}{d}\sum_{j=1}^d \sin^2k_j(x-c_1)\overset{j}{\check{\cdots}}(x-c_d). 
\end{align}

By Theorem~\ref{thm:main2}, any root of $P(x)$; $\mu\neq \pm 1$, is lifted up to the unit circle in the complex plain $e^{\pm i\arccos \mu}$ as the eigenvalues of $\hat{U}(\bs{k})$.
On the other hand, if $\pm 1$ are the root of $P(x)$ with the multiplicities $m_{\pm}$, then $\pm 1$ are also the eigenvalues of $\hat{U}(\bs{k})$ with the same multiplicities by Theorem~\ref{thm:main2}.  
Let us assume $\pm 1$ are not the root of $P(x)$. Then we obtain totally $2(d+1)$ eigenvalues of $\hat{U}(\bs{k})$, which is a contradiction because the dimension of the acting space of $\hat{U}(\bs{k})$ is $2d$. Thus at least $P(x)$ has $+1$ or $-1$ as the root.  Now to observe it more precisely, let us insert $x=1$ into (\ref{eq:ep}). 
Then we have 
\begin{align*}
    \det(1-\hat{T}(\bs{k})) &= (1-a_0)(1-c_1)\cdots(1-c_d)-\frac{1}{d}\sum_{j=1}^d (1-c_j^2)(1-c_1)\overset{j}{\check{\cdots}}(1-c_d)\\
    &= \frac{1}{d}(1-c_1)\cdots (1-c_d)\left\{d(1-a_0)-\sum_{j=1}^d (1+c_j)\right\} \\
    &= 0.
\end{align*} 
Here we used $\sin^2 k_j=1-c^2_j$ in the first equality and the definition of $a_0$ in the last equality, respectively. 
In the same way, inserting $x=-1$ into (\ref{eq:ep}), we obtain $\det(-1-\hat{T})=0$. 
Then the roots of $P(x)$ include both $+1$ and $-1$. 

Finally, let us proof the multiplicities of $x=\pm 1$. 
It is enough to clarify when the $r$-th derivative $(\partial/\partial x)^r \det(x-\hat{T}(\bs{k}))$ at $x=\pm1$ degenerates to $0$. To this end, let us start the derivative of (\ref{eq:ep}) and put $Q_j(x)=(x-c_1)\overset{j}{\check{\cdots}} (x-c_d)$.
Then we have 
\begin{align*}
    \frac{\partial}{\partial x}\det(x-\hat{T}(\bs{k}))\bigg|_{x=1} &= \frac{\partial}{\partial x}\left\{ (x-a_0)(x-c_1)\cdots (x-c_d) - \frac{1}{d}\sum_{j=1}^d \sin^2k_j(x-c_1)\overset{j}{\check{\cdots}}(x-c_d)\right\}\bigg|_{x=1} \\
    &= (1-c_1)\cdots(1-c_d)+(1-a_0)\sum_{j=1}^d Q_j(1)
    -\frac{1}{d}\sum_{j=1}^d (1+c_j)\sum_{\ell\neq j}Q_{\ell}(1) \\
    &= (1-c_1)\cdots(1-c_d)+(1-a_0)\sum_{j=1}^d Q_j(1)
    -\sum_{\ell=1}^d Q_{\ell}(1)\left((1-a_0)-\frac{1}{d}(1+c_\ell)\right) \\
    &= (1-c_1)\cdots(1-c_d)+\frac{1}{d}\sum_{\ell=1}^d Q_{\ell}(1)(1+c_\ell)\geq 0.
\end{align*}
The equality holds iff $Q_\ell(1)=0$ for any $1\leq \ell\leq d$. This condition is equivalent to $\#\{j\;|\; c_j=1\}\geq 2$. 
Then we have 
    \[\frac{\partial}{\partial x}\det(x-\hat{T}(\bs{k}))\bigg|_{x=1}= 0\]
iff $\bs{k}\in \delta F_s:=\{(k_1,\dots,k_d) \;|\;\#\{ j \;|\; c_j=1\}=s\}$ for $2\leq s\leq d$. This means the multiplicity $m_1$ is at least $2$. 
 
The multiplicity at $x=-1$ can be shown in a similar way.
\end{proof}

\begin{rem}
Let us consider the Grover walk with the flip flop shift on $\mathbb{Z}^d$. According to Proposition~3 in \cite{HKSS}, the eigenequaion of the time evolution in the Fourier space is described by 
    \[ (1-\lambda)^{d-1}(1+\lambda)^{d-1}\left(\lambda^2-\frac{2}{d}\sum_{j=1}^d\cos k_j\lambda+1\right), \]
if $(k_1,\dots,k_d)\notin (\mathbb{Z}\pi)^d$. 
This means that the discriminant operator coincides with the Fourier transform of the time evolution operator of the underlying random walk $(1/d)\sum_{j=1}^d \cos k_j$. 
The time evolution operator appears also as the $(1,1)$-element of the discriminant matrix of the moving shift QW in (\ref{eq:dom}). 
The multiplicities of $(\pm 1)$ are $(d-1)$ in the flip flop shift type while that are simple in the moving shift type. 
The density of the eigenspace which gives the localization in the real space for the flip flop shift type is higher than that for the moving shift type.  
Combining the result on \cite{KT}, we obtain that the eigenspaces for the localization are represented by each ``rectangle" of $\mathbb{Z}^d$ for the flip flop shift type while by each unit $d$-dimensional cube for the moving shift type in $\mathbb{Z}^d$. See Table~\ref{table:1}. 
\end{rem}

\begin{rem}
Let us see that the simplicity of $\pm 1$ solutions in (\ref{eq:EPT}) implies the no-existence of $\mathcal{L}^\perp$ which is defined in Sect.~\ref{sect:main} for almost all $\bs{k}\in \mathbb{T}^d$ except the boundary $\cup_{s\geq 2}\;\delta F_{s}^{(\pm)}$.  
Let $m_\pm$ be the multiplicities of the $\pm 1$ solutions of eigenpolynomial (\ref{eq:EPT}), respectively. 
Theorem~(\ref{thm:main2}) tells us that the eigenvalues except $\pm 1$ are lifted up to upper and down circles of the unit circle in the complex plane. Then one-eigenvalue between $(-1,1)\in\mathbb{R}$ of $\hat{T}$ generates two eigenvalues of $\hat{U}(\bs{k})$. On the other hand, the eigenvalues $\pm 1$ of $\hat{T}(\bs{k})$ are lifted up to those of  $\hat{U}(\bs{k})$ by a one-to-one map. Since the dimension of $\hat{U}(\bs{k})$ is $2d$ while the number of solutions of (\ref{eq:EPT}) are $d+1$, the dimension analysis leads to
\[ 2d\geq 2(d+1-m_+-m_-)+(m_+ + m_-) \]
which is equivalent to $m_+ + m_-\geq 2$. 
Here the inequality derives from a possibility of the existance of $\mathcal{L}^\perp$. 
However we have already known that $m_{\pm}=1$ by Proposition~\ref{pro:ddim} when $\bs{k}\notin \cup_{s\geq 2}\delta F_{s}^{(\pm)}$, then the equality holds. This is the reason for $\mathcal{L}^\perp =0$. 
Thus interestingly, considering the self adjoint operator on the star graph in the Fourier space denoted by (\ref{eq:dom}) is essentially same as considering the Grover walk on $\mathbb{Z}^d$ with the moving shift type. 
\end{rem}
\begin{table}[h]
    \centering
    \begin{tabular}{c|c|c}
         &  flip flop & moving shift \\ \hline 
    $d=1$& $0$ (free walk), $B_1$ & $2$ (zigzag walk), $S_1$\\
    $d=2$& $2$, $B_2$ & $2$, $S_2$ \\
    $d=3$& $4$, $B_3$ & $2$, $S_3$ \\
    $d=4$& $6$, $B_4$ & $2$, $S_4$ \\
    \vdots & \vdots & \vdots  
    \end{tabular}
    \caption{The multiplicities of $(\pm 1)$ for QWs with flip flop and moving shift types and the quotient graph. Here $B_d$ is the $d$-bouquet graph and $S_d$ is the star graph with $d$-claws and self loops for $\bs{k}\notin \cup_{s\geq 2}\delta F_{s}^{(\pm)}$. }
    \label{table:1}
\end{table}
\begin{rem}
Let us consider an exceptional case of $\bs{k}\in  \delta F_{d}^{(+)}$; that is, $\bs{k}=(0,\dots,0)$. 
The eigenpolynomial of  $\hat{T}(0,\dots,0)$ is reduced to 
\[ \det(x-\hat{T}(0,\dots,0))=(x-1)(x+1)^d. \]
Then we have 
\begin{align*}
    \dim(\ker(1-\hat{U}(0,\dots,0)|_{\mathcal{L}})) &=\dim(\ker(1-\hat{T}(0,\dots,0))= m_+=1, \\
    \dim(\ker(-1-\hat{U}(0,\dots0)|_{\mathcal{L}}))
    &=\dim(\ker(-1-\hat{T}(0,\dots,0))=m_-= d. 
\end{align*}
Note that $\kappa'=-1$ in Corollary~\ref{cor:dimensions}. Then Corollary~\ref{cor:dimensions} part 3 implies 
\begin{align*}
    \dim(\ker(-1-\hat{U}(0,\dots,0)|_{\mathcal{L}^\perp})) &=|E(B_d)|-p|V(B_d)|+m_-=d-(d+1)\times 1+d=d-1, \\
    \dim(\ker(1-\hat{U}(0,\dots,0)|_{\mathcal{L}^\perp})) &=|E(B_d)|-p|V(B_d)|+m_+=d-(d+1)\times 1+1=0. 
\end{align*}
Thus totally we have $\mathrm{Spec}(\hat{U}(0,\dots,0))=\{\{1\}^1,\{-1\}^{2d-1}\}$. 
Indeed, since 
\[\hat{U}(0,\dots,0)=\hat{S}(0)\sigma \mathrm{Gr}(2d)=\sigma^2 \mathrm{Gr}(2d)=\mathrm{Gr}(2d),\] 
the consistency can be confirmed by the spectrum of $2d$-dimensional Grover matrix. 
The $d-1$ eigenvalues exfoliate from the eigenvalue $-1$ by the splitting at the parameter $\bs{k}=(0,\dots,0)$~\cite{Kato1982}.   
\end{rem}
\section{Spectral information}
\subsection{Outline of the spectral information}
\begin{thm}\label{thm:zeta}
Let $G$ be a connected graph $n$ vertices and $m$ edges, $U=SC$ the the evolution matrix of a coined quantum walk on $G$. 
Suppose that $\mathrm{Spec} (C_u)= \{\kappa,\kappa' \}$. 
Set $1\leq  p\leq \delta(G)$. 
Assume $p= \dim \ker(\kappa-C_u)$ for any $u\in V$. 
Then, for the unitary matrix $U=SC$, we have   
\[
\det ( \lambda {\bf I}_{2|E|} -U)=(\lambda {}^2 -\kappa'^2 )^{|E|-p|V|} \det (\lambda {}^2 -(\kappa-\kappa')T\lambda-\kappa\kappa' ).  
\]
\end{thm}

\begin{proof}  Let us put $c=\kappa-\kappa'$, $m=|E|$ and $q=p|V|$. At first, we have  
\[
\begin{array}{rcl}
\  &   & \det ( {\bf I}_{2m} -z U)= \det ( {\bf I}_{2m} -z SC) \\ 
\  &   &                \\ 
\  & = & 
\det ({\bf I}_{2m} -z S((\kappa-\kappa') K \ K^* +\kappa' {\bf I}_{2m} )) \\ 
\  &   &                               \\ 
\  & = & 
\det ( {\bf I}_{2m} -\kappa'zS -cz SK \ K^* ) \\ 
\  &   &                \\ 
\  & = & \det ( {\bf I}_{2m} - cz SK \ K^* ( {\bf I}_{2m} -\kappa'zS )^{-1} ) \det ( {\bf I}_{2m} -\kappa'zS ) . 
\end{array}
\]

But, if ${\bf A}$ and ${\bf B}$ are an $m \times n $ and $n \times m$ 
matrices, respectively, then we have 
\[
\det ( {\bf I}_{m} - {\bf A} {\bf B} )= 
\det ( {\bf I}_n - {\bf B} {\bf A} ) . 
\]
Thus, we have 
\[
\det ( {\bf I}_{2m} -\kappa'zS )=(1-{\kappa'}^2 z^2 )^m . 
\]
Furthermore, we have 
\[
( {\bf I}_{2m} -\kappa'zS)^{-1} = \frac{1}{1-{\kappa'}^2 z^2 } ({\bf I}_{2m} +\kappa'zS ) . 
\]

Therefore, it follows that 
\[
\begin{array}{rcl}
\  &   & \det ( {\bf I}_{2m} -z U) \\ 
\  &   &                \\ 
\  & = & (1-{\kappa'}^2 z^2 )^{m} 
\det ( {\bf I}_{2m} - \frac{cz}{1-{\kappa'}^2 z^2 } SK \ K^* \ ( {\bf I}_{2m} +\kappa'zS )) \\ 
\  &   &                \\ 
\  & = & (1-{\kappa'}^2 z^2 )^{m} 
\det ( {\bf I}_{q} - \frac{cz}{1-{\kappa'}^2 z^2} K^* ( {\bf I}_{2m} +\kappa'zS) SK) \\ 
\  &   &                \\ 
\  & = & (1-{\kappa'}^2 z^2 )^{m} 
\det ( {\bf I}_{q} - \frac{cz}{1-{\kappa'}^2 z^2} K^* \ SK - \frac{\kappa'cz^2}{1-{\kappa'}^2 z^2 } K^* K ) \\ 
\  &   &                \\ 
\  & = & (1-{\kappa'}^2 z^2 )^{m-q} \det ((1-\kappa\kappa' z^2 ) {\bf I}_{q} -cz T ) . 
\end{array}
\]
Here we used Lemma~\ref{lem:123} in the last equality. 
 Let $u=1/ \lambda $. 
Then, we have 
\[
\det ({\bf I}_{2m} -1/  \lambda U)=(1-\kappa'^2 / \lambda {}^2 )^{|E|-q} \det ((1-\kappa\kappa'/ \lambda {}^2 ) {\bf I}_q - (\kappa-\kappa')/\lambda T) ,  
\]
and so, 
\[
\det ( \lambda  {\bf I}_{2m} -U)=( \lambda {}^2 -\kappa'^2 )^{m-q} \det (( \lambda {}^2 -\kappa\kappa') {\bf I}_q -(\kappa-\kappa') \lambda T) . 
\]
\end{proof}
 
Therefore from Theorem~\ref{thm:zeta}, an outline of how the eigenvalues of $U$ is obtained from $T$. 
\begin{cor}\label{cor:outlineeigensystem}
Let $G$ be a connected with $n$ vertices and $m$ edges. 
Then, the spectra of the unitary matrix $U=SC$ are given as follows: 
\begin{enumerate} 
\item $2p|V|$ eigenvalues: 
\begin{equation}\label{eq:specmap0}
\lambda = \frac{(\kappa-\kappa') \mu \pm \sqrt{ (\kappa-\kappa')^2 \mu {}^2 +4\kappa\kappa'}}{2} , \ \mu \in \mathrm{Spec} (T) ; 
\end{equation}
\item The rest eigenvalues are $\pm \kappa'$ with the same multiplicity $|E|-p|V|$. 
\end{enumerate} 
\end{cor}
\begin{proof}
By Theorem~\ref{thm:zeta}, we have 
\[
\begin{array}{rcl}
\  &   & \det ( \lambda  {\bf I}_{2|E|} -U) \\
\  &   &                \\ 
\  & = & ( \lambda {}^2 -\kappa'^2 )^{|E|-p|V|} 
\prod_{ \mu \in Spec (T)} ( \lambda {}^2 -(\kappa-\kappa') \mu \lambda -\kappa\kappa'), 
\end{array}
\]
and solving $ \lambda {}^2 - (\kappa-\kappa')\mu \lambda +1=0$, we obtain (\ref{eq:specmap0}). 
\end{proof}
We will see how this statement can be refined in general, in the next section. 
However we emphasise that the method treated in this section is quite useful to make a quick view of an outline of the spectral information of such kind of quantum walks.  

\subsection{Detailed spectral information}\label{sect:main}
Let $\mathcal{L}:= K \mathcal{V}_p+SK \mathcal{V}_p\subset \mathcal{A}$. 
We call this subspace the inherited subspace. The inherited subspace is the range of the map $L:\mathcal{V}_p\otimes \mathcal{V}_p \to \mathcal{A}$ such that $L(f\oplus g)=Kf+SKg$; that is, $\mathcal{L}=L(\mathcal{V}_p\oplus \mathcal{V}_p)$. 
\begin{lem}\label{lem:Lambda}
The subspace $\mathcal{L}$ is invariant under the action of $U$, that is, $U(\mathcal{L})=\mathcal{L}$. 
\end{lem}
\begin{proof}
It holds that $UK=\kappa SK$ and $USK=(\kappa-\kappa')SK T+\kappa' K$ by Lemma~\ref{lem:123}, and so we obtain 
\begin{align}\label{eq:key}
    UL=L\Lambda, 
\end{align}
where 
\begin{equation}\label{eq:Lambda}
    \Lambda=\begin{bmatrix} 0 & \kappa' \\ \kappa & (\kappa-\kappa') T \end{bmatrix}.
\end{equation} 
Then $U(\mathcal{L})\subset \mathcal{L}$. 
The inverse of $\Lambda$ exists. Indeed, 
\[\Lambda^{-1}=(-\kappa\kappa')^{-1}\begin{bmatrix} (\kappa-\kappa') T & -\kappa' \\ -\kappa & 0 \end{bmatrix}.\] Then $\mathcal{L}= U(\mathcal{L})$. 
\end{proof}

In this subsection, we see that more detailed spectral information. For example, let us give an example of a refinement of Corollary~\ref{cor:outlineeigensystem} as a consequence of this section. It is easy to check that the candidates of two eigenvalues are nominated from each eigenvalue $\pm 1\in \mathrm{Spec}(T)$ from the one-to-two map (\ref{eq:specmap0}). Such candidates for $+1$ (resp.$-1$) includes $-\kappa'$ (resp.$-\kappa'$). As we will see, the candidates $\pm \kappa'$ are rejected as the eigenvalue of $U|_\mathcal{L}$ unfortunately. 
However the rejected candidates are recovered as the eigenvalues of $U|_{\mathcal{L}^\perp}$. 
After all,  Corollary~\ref{cor:outlineeigensystem} can be refined by \begin{enumerate}
    \item $2p|V|-(m_+ + m_-)$ eigenvalues of $U|_{\mathcal{L}}$ are obtained from eigenvalues of $T$ except $\pm 1 \in \mathrm{Spec}(T)$ by the one-to-two map (\ref{eq:specmap0}) (or equivalently (\ref{eq:mapping})). 
    Here $m{_\pm}:=\dim\ker(\pm1-T)$. Due to such a singularity of $\pm 1\in \mathrm{Spec}(T)$, we have $\pm \kappa'\notin \mathrm{Spec}(U|_\mathcal{L})$ when $\kappa+\kappa'\neq 0$.    
    \item The eigenvalues of $U|_{\mathcal{\mathcal{L}^\perp}}$ are $\kappa'$ and $-\kappa'$ with $|E|-p|V|+m_\mp$ multiplicities, respectively.
\end{enumerate}

The following main result provides the information of eigenspace of $U$ and shows how to map not only eigenvalues but also the eigenspace of $T$ to the eigenspace $U|_{\mathcal{L}}$.   
Note that if $\kappa+\kappa'=0$, then the set of eliminated points $\{\pm \kappa'\}$ overlaps with $\{\pm \kappa\}$. This is the reason that we need the case study in the following theorem.  
\begin{thm}\label{thm:main2}
Let $G$ be a connected graph with $n$ vertices and $m$ edges, $U=SC$ the evolution matrix of a coined quantum walk on $G$. 
Suppose that $\mathrm{Spec} (C_u)= \{\kappa,\kappa' \}$ with $\kappa\neq \kappa'$.  
Set $1\leq  p\leq  \delta(G)$. 
Assume $p= \dim \ker(\kappa-C_u)$ for any $u\in V$.
Then we have 
\begin{enumerate}
\item $\kappa+\kappa'=0$ case
\begin{align}\label{eq:eigenspace1} 
    \ker(\lambda-U|_\mathcal{L}) 
    &= 
    \begin{cases}
    (1+{\kappa'}^{-1}\lambda S)K\ker\left(\;\frac{\mathrm{Im}[(\kappa \kappa')^{-1/2}\lambda]}{\mathrm{Im}[(\kappa/\kappa')^{1/2}]}-T\;\right) & \text{: $\lambda \notin \{\pm \kappa'\}$,}\\
    K\ker(\pm 1+T) & \text{: $\lambda \in  \{\pm \kappa'\}$.}
    \end{cases} 
\end{align}
\item $\kappa+\kappa'\neq 0$ case 
\begin{align}\label{eq:eigenspace2}
    \ker(\lambda-U|_\mathcal{L}) 
    =
    (1+{\kappa'}^{-1}\lambda S)K\ker\left(\;\frac{\mathrm{Im}[(\kappa \kappa')^{-1/2}\lambda]}{\mathrm{Im}[(\kappa/\kappa')^{1/2}]}-T\;\right).
\end{align}
In particular, $\ker(\pm\kappa'-U|_\mathcal{L})=0$. 
\end{enumerate}
On the other hand, $\mathrm{Spec}(U|_{\mathcal{L}^\perp}) = \{\pm \kappa'\}$ and 
    \[ \ker(\pm \kappa'-U|_{\mathcal{L}^\perp}) = \ker K^* \cap \ker(\pm 1 -S). \]
\end{thm}
\begin{proof}
First let us see 
\begin{equation}\label{eq:kerL}
    \ker (L)=\ker (\kappa'^2 -\Lambda^2).
\end{equation} 
For any $[f\;g]^\top\in \ker L$, it holds $Kf+SKg=0$ which implies 
\[\begin{bmatrix}f \\ g\end{bmatrix}\in \ker \begin{bmatrix} 1 & T \\ T & 1 \end{bmatrix} \] 
by Lemma~\ref{lem:123}. 
On the other hand, conversely, if $f+Tg=0$ and $Tf+g=0$, then 
\begin{equation}\label{eq:EV1}
    K^*(Kf+SKg)=0,\;K^*S(Kf+SKg)=0,
\end{equation}
which implies $Kf+SKg \in\mathcal{L}\cap \mathcal{L}^\perp=0$. 
Then we have 
    \[\ker L=\ker \begin{bmatrix} 1 & T \\ T & 1 \end{bmatrix}. \]
On the other hand, the Gaussian eliminations to $\kappa'^2-\Lambda^2$ using the expression of $\Lambda$ in (\ref{eq:Lambda}) gives  
    \[ \ker(\kappa'^2 -\Lambda^2) =\ker \begin{bmatrix} 1 & T \\ T & 1 \end{bmatrix}. \]
Then we have $\ker L = \ker(\kappa'^2 -\Lambda^2)$. 

Secondly, let us consider the following eigenequation; $U|_{\mathcal{L}}\psi=\lambda \psi$. 
Since $\psi\in \mathcal{L}$, there exists $\phi=[f\;g]^\top$ such that 
$\psi=L\phi$. By (\ref{eq:key}), we have 
\begin{align} 
U|_{\mathcal{L}}\psi =\lambda \psi 
&\Leftrightarrow L(\lambda -\Lambda)\phi=0,\;\phi\notin \ker L \notag \\
&\Leftrightarrow ({\kappa'}^2-\Lambda^2)(\lambda - \Lambda)\phi=0,\;\phi\notin \ker ({\kappa'}^2-\Lambda^2).\label{eq:iden}
\end{align}
This equation is the starting point for the spectral analysis on $U$. 
In the following let us consider the case for  $\lambda\neq \pm \kappa'$. Then 
(\ref{eq:iden}) is reduced to $(\lambda -\Lambda)\phi=0$.
Therefore $\phi=[f\;g]^\top$ satisfies 
\begin{align*}
    \begin{bmatrix}\lambda & -\kappa' \\ -\kappa & -(\kappa-\kappa')T+\lambda\end{bmatrix}\phi=0 
    & \Leftrightarrow 
    \begin{bmatrix}\lambda & -\kappa' \\ 0 & \lambda^2-(\kappa-\kappa')T\lambda -\kappa\kappa'\end{bmatrix}\phi =0
\end{align*}
by the Gaussian elimination. 
We can deform $\ker(-\lambda^2+(\kappa-\kappa')T\lambda +\kappa\kappa')$ by 
\[ \ker\left(\frac{\mathrm{Im}[(\kappa\kappa')^{-1/2}\lambda]}{\mathrm{Im}[(\kappa/\kappa')^{1/2}]}-T\right). \]
Therefore, $(\lambda-\Lambda)\phi=0$ with $\phi=[f\;g]^\top$ and $\lambda\neq \pm\kappa'$ if and only if 
$f=\kappa'\lambda^{-1} g$ and 
\[ g\in \ker\left(\frac{\mathrm{Im}[(\kappa\kappa')^{-1/2}\lambda]}{\mathrm{Im}[(\kappa/\kappa')^{1/2}]}-T\right). \] 
This gives the proof for the cases of both $\kappa+\kappa'=0$ and $\kappa+\kappa'\neq 0$ with $\lambda \notin \{\pm \kappa'\}$. 

Thirdly, let us consider the case for $\lambda= \pm \kappa'$. In this case, $\phi=[f\;g]^\top$ must satisfy that 
    \begin{equation}\label{eq:condition}
    (\pm\kappa'-\Lambda)^2\phi=0,\; \ker(\pm\kappa'-\Lambda)\phi\neq 0
    \end{equation}
by (\ref{eq:iden}). 
Direct computations of $\ker(\pm \kappa'-\Lambda)$ and $\ker(\pm \kappa'-\Lambda)^2$ from the explicit expression of $\Lambda$ in (\ref{eq:Lambda}) lead to 
\begin{align} 
\ker(\pm\kappa'-\Lambda)
&= \left\{ \begin{bmatrix} f \\ \pm f \end{bmatrix} \;\bigg|\; f\in \ker(\pm 1+T) \right\} \\
\label{eq:kerL2} \\ 
\ker(\pm\kappa'-\Lambda)^2
&=
\begin{cases}
\ker(\pm \kappa'-\Lambda) & \text{: $\kappa+\kappa'\neq 0$,} \\
\\ 
\left\{ \begin{bmatrix} f \\ g \end{bmatrix} \;\bigg|\; f,g\in \ker(\pm 1+T)\right\} & \text{: $\kappa+\kappa'=0$.}
\end{cases} \label{eq:jordan}
\end{align}
For any $f,g\in \ker(\pm 1+T)$, 
noting $Kg\pm SKg=0$ by (\ref{eq:EV1}), we have 
    \[ L\begin{bmatrix}f \\ g\end{bmatrix}
    =Kf + SKg= K(f \mp g).  \] 
Thus when $\kappa+\kappa'=0$, 
\[L(\ker(\pm\kappa'-\Lambda)^2\setminus \ker(\pm \kappa'-\Lambda)) = K\ker(\pm 1+T). \]
On the other hand, when $\kappa+\kappa'\neq 0$, then 
    \begin{equation}\label{eq:pm1}
    \pm \kappa' \notin \mathrm{Spec}(U|_\mathcal{L}). \end{equation}
which completes the proof of the spectral of $U|_\mathcal{L}$.

Finally, let us consider the case for $\psi\in \mathcal{L}^\perp=\ker(K^*)\cap \ker(K^*S)$. 
It holds 
\[\ker(K^*)\cap \ker(1\pm S) \subset \ker(K^*S)\]
Then we have 
    \[ \left(\ker(K^*)\cap \ker(1-S)\right) \oplus \left(\ker(K^*)\cap \ker(1+S)\right) \subset \ker(K^*)\cap \ker(K^*S) \]
On the other hand, since the inverse inclusion obviously 
holds, after all, we have 
\begin{align}
    \left(\ker(K^*)\cap \ker(1-S)\right) \oplus \left(\ker(K^*)\cap \ker(1+S)\right) = \ker(K^*)\cap \ker(K^*S)
\end{align}
For any $\psi\in \ker(K^*)\cap \ker(1-S)$, we have 
\[U\psi=S((\kappa-\kappa')KK^*+\kappa')\psi=\kappa'S\psi=\kappa'\psi,\]
while 
for any $\psi\in \ker(K^*)\cap \ker(1+S)$, we have 
\[U\psi=S((\kappa-\kappa')KK^*+\kappa')\psi=\kappa'S\psi=-\kappa'\psi.\]
Thus $\mathrm{Spec}(U|_{\mathcal{L}^\perp})=\{\pm \kappa'\}$ and $\ker(\pm \kappa'-U|_{\mathcal{L}^\perp})=\ker(\pm 1-S)$. 
Then we obtained all of the desired conclusion.
\end{proof}
\begin{cor}\label{cor:smt}(Spectral map) 
Let $\kappa=e^{i\xi}$, $\kappa'=e^{i\eta}$ and $\lambda=e^{i\phi}\in \mathrm{Spec}(U|_{\mathcal{L}})$. Then from Theorem~\ref{thm:main2}, if $\mu\in \mathrm{Spec}(T)$, then the following relation holds. 
\begin{equation}\label{eq:mapping}
    \sin[\phi-(\xi+\eta)/2]=\mu \sin[(\xi-\eta)/2].
\end{equation}
The geometric relation to show how the eigenvalues of $T$ is lifted up to the eigenvalues of $U|_{\mathcal{L}}$ is depicted by Figure~\ref{fig:1}.  
\end{cor}
\begin{rem}\label{rem:boundary}
If $\lambda=\pm\kappa$, 
    \[ \ker\left( \frac{\mathrm{Im}[(\kappa\kappa')^{-1/2}\lambda]}{\mathrm{Im}[(\kappa/\kappa')^{1/2}]}-T \right)=\ker(\pm 1-T), \]
while $\lambda=\pm \kappa'$,
    \[ \ker\left( \frac{\mathrm{Im}[(\kappa\kappa')^{-1/2}\lambda]}{\mathrm{Im}[(\kappa/\kappa')^{1/2}]}-T \right)=\ker(\mp 1-T). \]
By (\ref{eq:pm1}), when $\kappa+\kappa'\neq 0$, the eigenvalue $\pm 1$ of $T$ is lifted up to the eigenvalues of $U|_{\mathcal{L}}$ as not $\pm \kappa'$ but $\pm \kappa$. The eigenvector of $T$ for $\pm 1$ is mapped by $(1\pm \kappa'\kappa S)$. See also Figure \ref{fig:1}.  
\end{rem}

Let $\mathcal{O}:=K\mathcal{V}_p$ and $\mathcal{T}:=SK\mathcal{V}_p$. Then $\mathcal{L}=\mathcal{O}+\mathcal{T}$. 
We obtain the information of the dimensions as follows. 
\begin{cor}\label{cor:dimensions}(Dimension analysis) 
Let $m_{\pm}$ be the dimensions of $\ker(\pm1-T)$. 
Then we have 
\begin{enumerate}
    \item $\dim(\mathcal{O}+\mathcal{T})=\dim \mathcal{L}=2p|V|-(m_+ + m_-)$;
    \item $\dim(\mathcal{O}\cap \mathcal{T})=m_+ + m_-$; 
    \item $\dim(\mathcal{O}+\mathcal{T})^\perp= 2|E|-2p|V|+m_+ + m_-$; moreover $(\mathcal{O}+\mathcal{T})^\perp=\ker(\kappa'-U|_{\mathcal{L}^\perp})\oplus \ker(-\kappa'-U|_{\mathcal{L}^\perp})$ and 
    \[ \dim(\ker(\pm\kappa'-U|_{\mathcal{L}^\perp}))=|E|-p|V|+m_{\mp}. \]
\end{enumerate}
\end{cor}
\begin{rem}
Since $\dim(\ker(\pm\kappa'-U|_{\mathcal{L}^\perp}))\geq 0$, then 
    \[ m_{\pm}\geq p|V|-|E|  \]
by part (iii) of Corollary~\ref{cor:dimensions}. 
\end{rem}

Now we give the proof of Corollary~\ref{cor:dimensions}. 
\vskip\baselineskip
\noindent{\bf Proof of part 1.} \\
An eigenvalue of $T$ and its corresponding eigenvalues of $U|_{\mathcal{L}}$ mapped by (\ref{eq:mapping}) are one-to-two correspondence except $\pm1 \in \ker(\pm 1-T)$. The eingenvalues $\pm1 \in \ker(\pm 1-T)$ and its corresponding eigenvalues of $U|_{\mathcal{L}}$ are one-to-one by Remark~\ref{rem:boundary}.
Then since $T$ is a self adjoint operator, we have 
\begin{align*}
    \dim \mathcal{L} = 2\dim{\mathcal{V}_p}-(m_+ + m_-)=2=2p|V|-(m_+ + m_-).
\end{align*}\qed

\noindent {\bf Proof of part 2.}\\
The dimension of $\mathcal{L}$ can be computed by the other way as follows: 
\begin{align*}
    \dim \mathcal{L}&=\dim (K\mathcal{V}_p)+\dim(SK\mathcal{V}_p)-\dim(K\mathcal{V}_p \cap SK\mathcal{V}_p) \\
    &= 2p|V|-\dim(K\mathcal{V}_p \cap SK\mathcal{V}_p) \\
    &= 2p|V|-(m_+ + m_-).
\end{align*}
Here the second equality derives from that $K$ is an  injection because of $K^*K=I$ in the statement of Lemma~\ref{lem:123} and the third equality is obtained by part (i) of Corollary \ref{cor:dimensions}.
Then we obtain the desired conclusion. \qed
\vskip\baselineskip
\noindent{\bf Proof of part 3.}\\
First let us consider the relation between two eigenvalues $\lambda_\mu^{(1)}$ and $\lambda_\mu^{(2)}$ of $U|_{\mathcal{}L}$ obtained by $\mu\in \mathrm{Spec}(T)$ the eigenvalue. 
From (\ref{eq:mapping}), the eigenvalues $\lambda_\mu^{(1,2)}$ are the solution of 
    \begin{equation}\label{eq:mappingquad} 
    \lambda^2 - (\kappa-\kappa)\mu \lambda -\kappa \kappa'=0. 
    \end{equation}
Then we have 
\begin{equation}\label{eq:kaitokeisuu}
    \lambda_\mu^{(1)}\lambda_\mu^{(2)}=-\kappa\kappa'.
\end{equation}
Note that there exists $\mu\in[-1,1]$ such that $\lambda_\mu^{(1)}=\lambda_\mu^{(2)}$ only if $\kappa+\kappa'=0$.
Let us consider the case for $\kappa+\kappa'\neq 0$ first. 
Assume $\mu\notin\{\pm 1\}$. This is equivalent to 
\begin{equation}
    \lambda_\mu^{(1,2)}\notin\{\pm \kappa, \pm\kappa'\}
\end{equation} by (\ref{eq:mappingquad}).  
Let $g\in \ker(\mu -T)$. Then the eigenvectors lifted up to the ones of $U|_{\mathcal{L}}$ are 
\[ \psi_1=(1+\kappa'^{-1}\lambda_\mu^{(1)}S)Kg \text{\;and\;} \psi_2=(1+\kappa'^{-1}\lambda_\mu^{(2)}S)Kg. \]
by Theorem~\ref{thm:main2}. 
The subspace inherited by $\ker(\mu-T)$ is denoted by \[\mathcal{L}_\mu:=(1+\kappa^{-1}\lambda_mu^{(1)}S)K\ker(\mu-T). \]
We find a linear combination of $\phi_1$ and $\phi_2$ so that 
    \[\varphi_g=x\psi_1+y\psi_2\in \ker(1-S) \]
with $x,y\neq 0$. 
Indeed, if $x=\kappa'^{-1}\lambda_\mu^{(2)}-1$ and $y=-\kappa'^{-1}\lambda_\mu^{(1)}+1$, then $\varphi_g=\kappa'^{-1}(\lambda^{(2)}_\mu-\lambda^{(1)}_\mu)(1+S)Kg$, which is $\varphi_g\in \ker(1-S)$. 
In the same way, 
    \[\varphi'_g=x'\psi_1+y'\psi_2\in \ker(1+S)\]
with $x',y'\neq 0$, where 
    $x=\kappa'^{-1}\lambda_\mu^{(2)}+1$ and $y=-(\kappa'^{-1}\lambda_\mu^{(1)}+1)$. 
Let $\{g_1,g_2,\dots,g_s\}$ be CONS of $\ker(\mu-T)$. 
Then $\langle \varphi_{g_i}, \varphi_{g_j} \rangle=\langle \varphi_{g_i}', \varphi_{g_j}' \rangle=0$ if $i\neq j$. 
This means $\mathcal{L}_\mu\cap \ker(1-S)=\mathrm{span}(\varphi_{g_j})_{j=1}^{s}$ and $\mathcal{L}_\mu\cap \ker(1+S)=\mathrm{span}(\varphi_{g_j}')_{j=1}^{s}$. 
Then we have 
    \begin{equation}\label{eq:munotpm1}
        \dim(\mathcal{L}_\mu \cap \ker(\pm1-S))=\dim(\mu-T)
    \end{equation}
if $\mu\neq\pm1$. 
In the next, let us consider $\mu\in \{\pm 1\}$. 
By (\ref{eq:mappingquad}), we have 
    \[ \{\lambda_\mu^{(1)},\lambda_\mu^{(2)}\}=\begin{cases} \{\kappa,-\kappa'\} & \text{: $\mu=1$,} \\
    \{-\kappa,\kappa'\} & \text{: $\mu=-1$.}\end{cases} \]
Note that by Remark~\ref{rem:boundary}, the values $\pm \kappa'$ must be omitted as the eigenvalues of $U|_\mathcal{L}$. 
So if $\mu=1$, then the eigenvector which can be lifted up to the eigenvector of $U|_{\mathcal{L}}$ from $\varphi=g\in \ker(1-T)$ is only $(1+\kappa\kappa'S)Kg$. 
Let us see when $g\in \ker(1-T)$, $Kg=SKg$: 
it holds that 
    \begin{align*}
        g\in \ker(1-T) 
        & \Rightarrow K^*SKg-g=0 \\
        & \Rightarrow K^*(SKg-Kg)=0,\; K^*S(Kg-SKg)=0 \\
        & \Rightarrow SKg-Kg\in \mathcal{L} \cap \mathcal{L}^\perp=\{0\}.
    \end{align*}
Then we have $Kg\in \ker(1-S)$, that is, $(1+S)/2\;Kg=Kg$. 
Therefore 
    \begin{align*}
        \varphi=(1+\kappa'^{-1}KS)Kg &= (1+\kappa'^{-1}KS)\frac{1+S}{2} Kg \\
        &= \frac{1+\kappa'^{-1}\kappa}{2}(1+S)Kg\in \ker(1-S). 
    \end{align*}
This implies 
    \begin{align}
        \dim \Lambda_1=\dim \Lambda_1\cap \ker(1-S) &=\dim\ker(1-T)=m_+ \label{eq:1},\\
        \dim \Lambda_1\cap \ker(1+S) = 0. \label{eq:1'}
    \end{align}
In the same way for $\mu=-1$, we have 
    \begin{align}
        \dim \Lambda_{-1}=\dim \Lambda_{-1}\cap \ker(1+S) &=\dim\ker(-1-T)=m_{-}, \label{eq:2}\\
        \dim \Lambda_1\cap \ker(1+S) = 0. \label{eq:2'}
    \end{align}
All up together with (\ref{eq:munotpm1}), (\ref{eq:1}), (\ref{eq:1'}) and (\ref{eq:2}), (\ref{eq:2'}), we finally obtain
\begin{align}
\dim (\mathcal{L}\cap \ker(1-S)) &= \sum_{\mu\in \mathrm{Spec}(T)}\dim(\mathcal{L}_\mu\cap \ker(1-S)) \notag\\
&= \sum_{\mu\in\mathrm{Spec}(T)\setminus\{-1\}}\dim(\ker(\mu-T)) \notag\\
&= p|V|-m_{-} \\
\dim (\mathcal{L}\cap \ker(1-S)) 
&= p|V|-m_{+}.
\end{align}
Therefore 
\begin{align*}
    \dim (\mathcal{L}^\perp\cap \ker (1-S)) &= \dim\ker(1-S)-(p|V|-m_{-})=|E|-p|V|+m_-, \\
    \dim (\mathcal{L}^\perp\cap \ker (1+S)) &= \dim\ker(1+S)-(p|V|-m_{+})=|E|-p|V|+m_+.
\end{align*}
The case for $\kappa+\kappa'=0$ can be proved by almost the same arguments.\qed
\subsection{Twisted quantum walk}\label{subsect:oneform}
Let us put a one form $\theta: A\to \mathbb{R}$ such that $\theta(\bar{a})=-\theta(a)$. We define the twisted shift operator $S_\theta$ on $\ell^2(A)$ by $(S_\theta\psi)(a)=e^{i\theta(a)}\psi(\bar{a})$. 
When $\theta=0$, the usual flip flop shift operator is reproduced. We consider the quantum walk by $U_\theta=S_\theta C$. This quantum walk with such a one form $\theta$ is useful to observe the spectral information on the crystal lattices.   
It is easy to check that the properties of the shift operator $S$ which are used all the proofs in this section are only (i) $\mathrm{Spec}(S)=\{\pm 1\}$ and (ii) $S$ is a unitary operator. An operator satisfying (i) and (ii) is called an involution. In particular, we used the fact that the square of an  involution operator becomes the identity operator. 
Since the operator $S_\theta$ is still an involution for any $\theta$, it is easy to check that all the results on this section holds just changing $S\to S_\theta$ and $T\to T_\theta:=KS_\theta K^*$. We give an example of an application of this fact in Section~\ref{sect:example} for the $\mathbb{Z}^d$ lattice. 
\section{Summary and future's work}
We considered coined quantum walks on graphs whose local coins $C_u$ satisfy
\begin{enumerate}
    \item $\mathrm{Spec}(C_u)=\{ \kappa, \kappa'\}$ ($\kappa\neq \kappa'$),
    \item $\dim(\ker(\kappa-C_u))=p$,
\end{enumerate}
where $1\leq p\leq \delta (G)$. 
Here if $p=1$ and $\kappa=1,\;\kappa'=-1$, then the walk recovers a Szegedy walk. 
We showed that the time evolution operator $U$ is spectral decomposed into $U=U|_{\mathcal{L}}\oplus U|_{\mathcal{L}^\perp}$. We fined the self adjoint operator $T$ on $\ell^2(V;\mathbb{C}^p)$ whose eigenvalues are lifted up as the eigenvalues and the eigenspace of $U|_\mathcal{L}$ by the one-to-two mapping (\ref{eq:mapping}) except $\pm 1\in \mathrm{Spec}(T)$. 
We also showed how the eigenspace of $U|_{\mathcal{L}}$ are obtained from the eigenspace of $T$ by (\ref{eq:eigenspace1}), (\ref{eq:eigenspace2}).  
We clarified that the values $\pm \kappa$ and $\pm \kappa'$ are nominated as the eigenvalues of $U|_\mathcal{L}$ by the mapping (\ref{eq:mapping}), while $\{\pm \kappa'\}$ are not selected as the eigenvalues of $U|_{\mathcal{L}}$. 
But the rejected values $\pm \kappa'$ by the selection of the $\mathrm{Spec}(U|_{\mathcal{L}})$ recovers as the eigenvalues of $U|_{\mathcal{L}^\perp}$. 
Thus totally the eigenpolynomial (\ref{eq:specmap0}) holds. 

After all, interestingly, we confirmed that the spectral structure on this class of quantum walks essentially keeps the spectral structure in the case of $p=1$. 
Only the difference is the shape of the discriminant operator $T$. In our setting, the domain of $T$ is $\ell^2(V;\mathbb{C}^p)$. We gave examples for $\mathbb{Z}^d$ with the moving shift. The discriminant operator $T$ can be regarded as an extension of the double stochastic matrix with some condition. 
In any case, the problem can be reduced to the study on $T$. If $p=1$, then the problem is reduced that of the random walk and we can use well developed facts on random walks if $T$ is reversible, while $p\geq 2$, the reduced ``walk" has an internal degree of freedom $p$ and it seems that there are small insights into $T$ comparing with the random walks with the huge insights in the present stage. Therefore exploring study on $T$ is one of the interesting future's problems.

\begin{figure}[htbp]
 \centering
 \includegraphics[keepaspectratio, width=15cm]
      {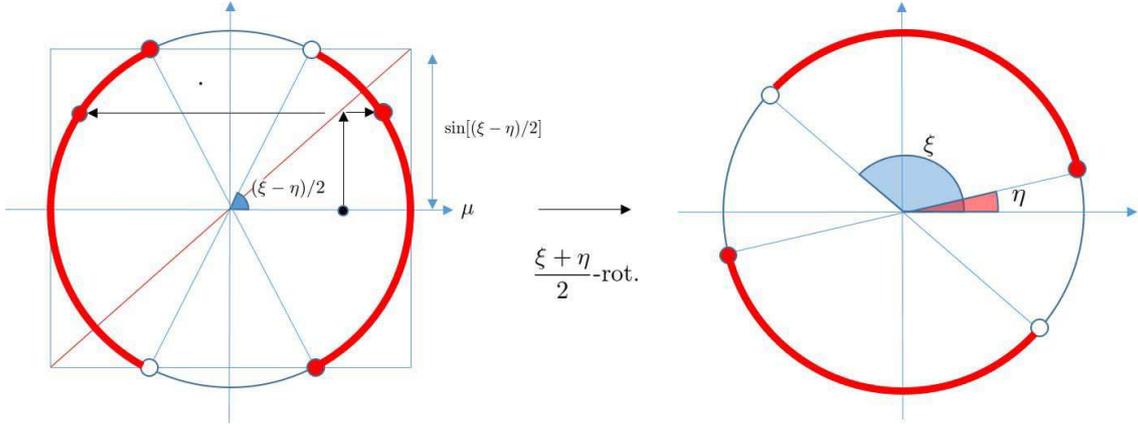}
 \caption{The way of the mapping from $\mathrm{Spec}(T)$ to $\mathrm{Spec}(U_{\mathcal{L}})$: The left figure shows how each eigenvalue of $T$; $\mu$, is mapped on the unit circle as $e^{\phi}e^{-(\xi-\eta)/2}$, where $e^{i\phi}$ is the eigenvalue of $U|_\mathcal{L}$,  $e^{i\xi}=\kappa$ and $e^{i\eta}=\kappa'$. Each eigenvalue of $T$ except $\pm 1$ are mapped to two points of the unit circle, and each eigenvalue of $\pm 1$ is mapped to only one points. The right figure is obtained by the $(\xi+\eta)/2$ rotation of the left figure, and this shows the region including the support of the eigenvalues of $U|_{\mathcal{L}}$. The white two points on the unit circle, which represents the blanks as the support of the $\mathrm{Spec}(U|_\mathcal{L})$, are filled by the eigenvalues of $\mathcal{U}|_{\mathcal{L}^\perp}$ with the multiplicities $|E|-p|V|+m_\pm$. }
 \label{fig:1}
\end{figure}

\vskip\baselineskip
\noindent{\bf Acknowledgments}: E.S. acknowledges financial supports from the Grant-in-Aid of
Scientific Research (C) Japan Society for the Promotion of Science (Grant No.~19K03616) and Research Origin for Dressed Photon. Y.S. thanks the financial supports from JSPS KAKENHI (Grant No. 19H05156) and JST PRESTO (Grant No. JPMJPR20M4).

\begin{small}
\bibliographystyle{jplain}

\end{small}

\end{document}